\documentclass[12pt, draftclsnofoot, onecolumn]{IEEEtran}

\usepackage{amsfonts}
\usepackage{slashbox}
\usepackage{amsmath}
\usepackage{multirow}
\usepackage{graphicx}
\usepackage{amsfonts}
\usepackage{amsmath,epsfig}
\usepackage{mathbbold}
\usepackage{stfloats}
\usepackage{amssymb}
\usepackage{url}
\usepackage{bm}
\usepackage{cite}
\usepackage{amsmath}
\usepackage{amssymb}
\usepackage{latexsym}
\usepackage{multirow}
\usepackage{epsfig}
\usepackage{graphics}
\usepackage{mathrsfs}
\usepackage{algorithmic}
\usepackage{algorithm}
\usepackage{subfigure}
\usepackage{slashbox}
\usepackage[all]{xy}

\usepackage{pifont}
\usepackage{bbding}
\usepackage{stmaryrd}
\usepackage{amssymb}
\usepackage{amsfonts}
\usepackage{epic}
\usepackage{stfloats}
\usepackage{latexsym}
\usepackage{epstopdf}
\usepackage{epic}
\usepackage{bm}
\usepackage{xcolor}
\usepackage{mathrsfs}
\usepackage{pifont}
\usepackage{bbding}
\usepackage{amsmath,epsfig}
\usepackage{mathbbold}
\usepackage{stmaryrd}
\usepackage{amssymb}
\usepackage{amsfonts}
\usepackage{epic}
\usepackage{graphicx}
\usepackage{subfigure}

\usepackage{enumerate}

\usepackage{stfloats}
\usepackage{latexsym}
\usepackage{epstopdf}
\usepackage{epic}
\usepackage{multirow}
\usepackage{stfloats}
\usepackage{bm}

\usepackage{booktabs}
\usepackage{color}
\usepackage{cases}
\usepackage{amsthm}

\newtheorem{theorem}{\underline{Theorem}}
\newtheorem{lemma}{\underline{Lemma}}

\newcommand{\q}{\mathbf q}
\newcommand{\Q}{\mathbf Q}
\newcommand{\A}{\mathbf A}
\newcommand{\w}{\mathbf w}
\newcommand{\pow}{\mathbf P}
\newcommand{\cc}{\mathbf c}
\newcommand{\La}{\bm{\lambda}}
\newcommand{\Mu}{\bm{\mu}}
\newcommand{\Be}{\bm{\beta}}
\newcommand{\Nu}{\bm{\nu}}
\newcommand{\N}{\mathcal{N}}
\newcommand{\K}{\mathcal{K}}

\begin{document}

\title{{Common Throughput Maximization in UAV-Enabled OFDMA Systems with Delay Consideration}}
\author{\IEEEauthorblockN{Qingqing Wu,  \emph{Member, IEEE},  and Rui Zhang, \emph{Fellow, IEEE}
\thanks{ The authors are with the Department of Electrical and Computer Engineering, National University of Singapore, email:\{elewuqq, elezhang\}@nus.edu.sg. Part of this work has been presented in IEEE APCC 2017 as an invited paper \cite{wu2017ofdma}.  }} }

\maketitle
%

\begin{abstract}
The use of unmanned aerial vehicles (UAVs) as communication platforms is of great practical significance in future wireless networks, especially for on-demand deployment in temporary events and emergency situations. Although prior works have shown the performance improvement by exploiting the UAV's mobility, they mainly focus on delay-tolerant applications.
As delay requirements fundamentally limit the UAV's mobility, it remains unknown whether the UAV is able to provide any performance gain in delay-constrained communication scenarios.  Motivated by this, we study in this paper a UAV-enabled orthogonal frequency division multiple access (OFDMA) network where a UAV is dispatched as a mobile base station (BS) to serve a group of users on the ground. We consider a minimum-rate ratio (MRR) for each user, defined as the minimum instantaneous rate required over the average achievable throughput, to flexibly adjust the percentage of its delay-constrained data traffic.
Under a given set of constraints on the users' MRRs, we aim to maximize the minimum average throughput of all users by jointly optimizing the UAV trajectory and OFDMA resource allocation. First, we show that the max-min throughput in general decreases as the users' MRR constraints become more stringent, which reveals a fundamental throughput-delay tradeoff in UAV-enabled communications.  Next,  we propose an iterative parameter-assisted block coordinate descent method to optimize the UAV trajectory and OFDMA resource allocation alternately,  by applying the successive convex optimization and the Lagrange duality, respectively. Furthermore, an efficient and systematic UAV trajectory initialization scheme is proposed based on a simple circular trajectory. Finally, simulation results are provided to verify our theoretical findings and demonstrate the effectiveness of our proposed designs.
\end{abstract}


\section{Introduction}
As predicted in \cite{drone_num},  the number of unmanned aerial vehicles (UAVs), also known as drones, will continue to surge over the next several years, with the global annual unit shipment increasing more than tenfold from 6.4 million in 2015 to 67.7 million by 2021. Such a dramatic increase is due to the steadily decreasing cost of UAVs and their fast-growing  demand in applications such as surveillance and monitoring, aerial camera and radar, cargo delivery, communication platforms, etc. In June 2017,  the International Olympic Committee (IOC) and Intel announced to use UAVs to enhance the future Olympic Games experience, such as the drone light show  \cite{Intel_drone}.
Just in the same month, the ``Safe DRONE Act of 2017'' was proposed to the United States Congress for accelerating the development of the UAV technology, which requested 14 million dollars for funding a host of research projects on UAVs \cite{ACT2017_drone}.  In fact, many prominent companies such as Google, Qualcomm, Amazon, and Nokia,  have already launched their respective programs  to advance the UAV research and conduct UAV field tests. In particular, the recent trial results released by Qualcomm have shown that the current fourth-generation (4G) cellular network can provide reliable communications for UAVs at an altitude up to 400 feet  \cite{qualcom_UAVreport}, which paves the way for realizing cellular-enabled UAV communications in future.
Besides being ``drone clients'' of wireless networks, UAVs can also be employed as various aerial communication platforms, to help improve the performance of existing wireless communication systems such as cellular networks \cite{zeng2016wireless}.

Compared to the traditional terrestrial wireless communications, UAV-enabled communications mainly have the following three advantages. First, they are significantly   less affected by channel impairments such as shadowing and fading, and in general possess more reliable air-to-ground channels due to the much higher possibility of having line-of-sight (LoS) links with ground users. Second, UAVs can be deployed more flexibly and moved more freely in the three-dimensional (3D) space. Third, the high mobility of UAVs can be fully controlled to enhance desired communication links and/or avoid undesired interference via proper trajectory design. These features
bring both opportunities and challenges in designing UAV-enabled wireless communications,
which were not explored before in conventional terrestrial systems with fixed ground base stations (BSs) \cite{zeng2016wireless}.  As UAVs are suitable to serve as  airborne communication hubs such as aerial BSs and/or relays, they are especially useful for the practical scenarios that require  on-demand deployment in temporary events or emergency situations (such as natural disaster), when the ground infrastructures are insufficient or even unavailable.
However, there are also open  challenges that need to be tackled before realizing the full potential of UAV-enabled communications,  such as UAV deployment and trajectory design, communications resource allocation and multiple access, etc.

 Recently, UAV deployment problems have been extensively studied in the literature \cite{al2014optimal,lyu2016placement,bor2016efficient,chetlur2017downlink2,fotouhi2016dynamic,zhang2017spectrum}. The main design objective  is to optimize the UAVs' altitude, horizontal positions, and/or spatial  density to achieve the maximum communication coverage in a given area. However, the high mobility or trajectory design of UAVs is not considered in these works.  An energy-efficient relaying scheme is proposed in \cite{li2016energy} where multiple UAVs cooperatively relay data packets from ground sensors to a remote BS based on time division multiple access (TDMA). Although  the UAV mobility is considered in  \cite{li2016energy}, the UAVs' trajectories are assumed to be pre-determined and not optimized, which simplifies the design to a UAV-packet matching problem.   In \cite{mozaffari2016unmanned}, both the static and mobile UAV-enabled wireless networks are studied which are underlaid with  a device-to-device (D2D) communication network.
Yet, the mobile UAV is only allowed to communicate  at a set of  stop points. As a result, the UAV trajectory is highly restricted and does not fully exploit the UAV's  high mobility for performance optimization.
As such, a joint UAV trajectory and adaptive communication design is more promising.  
Motivated by this, a general  trajectory and communication joint  optimization framework is proposed in \cite{zeng2016throughput} for a UAV-enabled mobile relaying system, which is also extended to the energy efficiency maximization in a point-to-point UAV-ground communication system  \cite{zeng2016energy}.
The UAV trajectory designs in  \cite{zeng2016throughput} and \cite{zeng2016energy} can be considered as a generalization of the UAV deployment problem, subject to practical constraints on the UAV's mobility, such as its initial/final locations, maximum speed and acceleration, etc. Thus, it is intuitive that a mobile UAV with optimized trajectory in general can achieve  higher throughput than a static UAV, even with optimal deployment/placement.
For UAV-enabled multiuser communication networks, a novel cyclical multiple access scheme is proposed in \cite{lyu2016cyclical} where the UAV periodically serves each of the ground users along its cyclical trajectory via TDMA.
In \cite{JR:wu2017joint}, a joint user scheduling, power control, and trajectory optimization problem is investigated for a multi-UAV enabled multiuser  system. It is shown that with joint trajectory design and power control, the UAVs can cooperatively serve the ground users in a periodic manner with strong LoS links and yet   avoiding severe interference. However, this is achieved by scheduling each user to communicate only when its associated UAV is
sufficiently close to it in each UAV flight period, which implies that the user has to wait for the next UAV flight period to communicate again.
As a result, the throughput gain brought by the UAV's mobility does not come for free but in fact at the cost of user communication delay.

Future wireless networks are expected to provide different quality-of-service (QoS) guarantees for a wide range of applications with diversified requirements \cite{wu2016overview,zhang2016fundamental}. In fact, the delay requirements of wireless multimedia services may vary dramatically in a large scale from milliseconds such as for video conferencing and online gaming, to several seconds  for file downloading/sharing and data backup. For example, when a user is downloading files which are delay-tolerant in general,  some other users in the same area may be watching high-definition (HD) movies on YouTube that require minimum rates at any time.  In light of this, wireless networks with such mixed services should be optimized to not only maximize the system total throughput but also meet the heterogeneous delay requirements of different applications. Although the UAV trajectory design has been shown  to significantly enhance the throughput of various  wireless communication systems such as  for mobile relaying channel \cite{zeng2016throughput,li2016energy}, multiple access channel (MAC) and broadcast channel (BC) \cite{wu2017joint,JR:wu2017_capacity,lyu2016cyclical}, interference channel (IFC)\cite{JR:wu2017joint,mozaffari2016unmanned}, and wiretap channel \cite{guangchi2016_UAV,mozaffari2016unmanned},  all these works consider only delay-tolerant applications where the users' instantaneous rates in general cannot be guaranteed. Therefore, it remains unknown whether mobile UAVs are still able to provide throughput gains over static UAVs/ground BSs, when the users'  delay or minimum-rate requirements are considered.


    \begin{figure}[!t]
\centering
\includegraphics[width=0.5\textwidth]{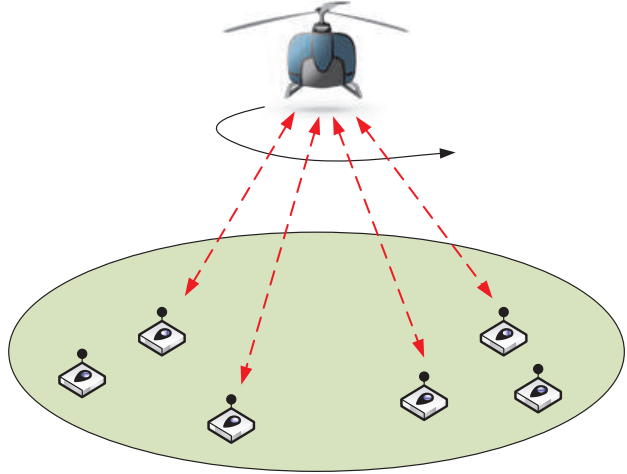}
\caption{ A UAV-enabled wireless network.} \label{UAV}\vspace{-0.4cm}
\end{figure}


Motivated by such an open question, we study in this paper a UAV-enabled orthogonal frequency division multiple access (OFDMA) system, where a UAV is dispatched as a mobile BS to serve a group of users on the ground during  a given finite period, as shown in Fig. 1. Besides having been standardized as the downlink multiple access scheme in the current 4G networks,  OFDMA is also deemed as a promising candidate  for the forthcoming fifth-generation (5G) wireless networks \cite{andrews2014will,qing1,wang2017joint}.
 This is essentially attributed to its various advantages, e.g., flexible bandwidth and power allocation over users, which fits particularly well in our considered  UAV-enabled network with heterogeneous user delay requirements. {Without loss of generality, we consider two types of data traffic for each user: delay-constrained traffic where a minimum rate should be supported at any time during the UAV service period versus delay-tolerant traffic where no minimum rate is required and the transmission rate can be elastically allocated over the period. To this end, we introduce a minimum-rate ratio (MRR) for each user which is defined as the minimum  instantaneous rate required over the achievable average throughput of the whole period. As such,  the MRR can be flexibly adjusted by each user to specify the percentage of the delay-constrained traffic required versus that of the delay-tolerant traffic, depending on real-time applications.}  Our goal is to maximize the minimum average throughput of all  ground users while meeting a given set of constraints on the users' MRRs,  by jointly optimizing the UAV trajectory and OFDMA resource allocation.
  Compared to prior works \cite{li2016energy,zeng2016throughput,zeng2016energy,jeong2016mobile,lyu2016cyclical,JR:wu2017joint,wu2017joint,mozaffari2016unmanned,Mozaffari2016}, such a joint UAV trajectory and resource allocation design is more general and practically useful  since the user communication delay requirements are  taken into account.
 This thus leads to a more fair throughput comparison with the terrestrial BSs or static UAVs and also helps to reveal a better understanding of the fundamental throughput  gain achievable  by exploiting the UAV's mobility control subject to delay constraints. Intuitively, when the UAV flies towards some users to capture better channels with them, it  gets farther away spontaneously from other users that are not in its heading direction and thus experiences degraded channels. As a result, more bandwidth and transmit power need to be allocated to those users so that their minimum rates can be achieved. This in turn would limit the potential rate increase  of the users in the UAV's heading direction.  Thus, there exists a new and non-trivial tradeoff in the UAV trajectory design for throughput maximization when minimum-rate or delay constraints are considered.


 The main contributions of this paper are summarized as follows. First, we formulate a joint UAV trajectory and OFDMA resource allocation optimization problem to maximize the minimum throughput  of ground users while guaranteeing their specified MRR constraints. Next, we show that the system max-min throughput is non-increasing with respect to the users' MRRs, which implies that the throughput gain of mobile UAVs over static UAVs reduces as the delay constraints become more stringent.
 Although the formulated problem is non-convex and challenging to solve, we propose an efficient iterative block coordinate descent algorithm to solve the bandwidth and power allocation problem and the UAV trajectory optimization problem alternately, which is guaranteed to converge. Specifically, in each iteration, the bandwidth and power allocation problem is solved optimally by applying the Lagrange duality method with given fixed UAV trajectory.
 While for the UAV trajectory optimization problem with fixed OFDMA resource allocation, the successive convex optimization technique is applied to tackle its non-convexity. However, due to the MRR constraints, it is shown that the conventional block coordinate descent method which directly iterates between the OFDMA resource allocation and UAV trajectory design will very likely get stuck at the initial point, which leads to an ineffective update of the UAV trajectory. To overcome this issue, we propose a new {\it parameter-assisted} block coordinate descent method where each parameter is a temporary MRR set to be larger than the target MRR for a corresponding user.
 Then at each iteration, we gradually decrease the temporary MRRs before solving the UAV trajectory optimization problem, until they reach the target MRRs for all users.  It is shown that this new method can effectively update the UAV trajectory in each iteration, thus resolving the issue of conventional block coordinate descent method.  Furthermore, we propose a systematic and low-complexity UAV trajectory initialization scheme based on the simple circular trajectory. Finally, numerical results are provided to verify  the fundamental tradeoff between the system max-min throughput and the users' delay/MRR constraints  and demonstrate the effectiveness of our proposed designs.

The rest of this paper is organized as follows. Section II introduces the system model and the problem formulation for a UAV-enabled OFDMA network.
In Section III, we propose an efficient iterative algorithm as well as a general UAV trajectory initialization scheme for the considered problem.
Section VI presents the numerical results to demonstrate the performance of the proposed designs. Finally, we conclude the paper  in Section VI.

\emph{Notations:} In this paper, scalars are denoted by italic letters, while vectors and matrices are respectively denoted by bold-face lower-case and upper-case letters. $\mathbb{R}^{M\times 1}$ denotes the space of $M$-dimensional real-valued vector. For a vector $\mathbf{a}$, $\|\mathbf{a}\|$ represents its Euclidean norm, $\mathbf{a}^T$ denotes its transpose, and  $\mathbf{a}\succeq 0$ indicates that $\mathbf{a}$ is element-wise larger than or equal to 0. For a time-dependent function $\mathbf{x}(t)$,  $\dot{\mathbf{x}}(t)$ denotes its derivative with respect to time $t$. For a set $\mathcal{K}$, $|\mathcal{K}|$ denotes its cardinality.

\section{System Model and Problem Formulation}

\subsection{System Model}
As shown in Fig. 1, we consider a UAV-enabled OFDMA system where the UAV is employed as an aerial BS to serve a group of $K$ users on the ground.  The user set is denoted by $\mathcal{K}$ with $|\mathcal{K}|=K$.
  In practice, the users that need to be served by the UAV can be either the terrestrial BSs that have no ground backhaul links or the ground mobile terminals in a geographical area that is poorly or not even covered by existing terrestrial BSs.
 At any time during the UAV flight period, denoted by $T$,  the UAV can communicate with multiple ground users simultaneously by employing OFDMA, i.e., assigning each user a fraction of the total bandwidth/transmit power.
 In general,  from the perspective of throughput maximization, it is intuitive that larger $T$ is desirable since it will provide the UAV more time to fly closer to each ground user, leading to better air-to-ground links. The effect of  $T$  on the system performance will be investigated  in detail in Section IV. {Since we target for a centralized  design that is implemented off-line, the proposed algorithm can be performed at a central controller (e.g., installed on the UAV)  who is able to collect all the users' information such as their  locations and MRRs. Then the obtained solutions can be programmed into the control and communication units  of the UAV.}


 We assume that the horizontal coordinate of each ground user is known in advance and fixed at ${\mathbf{w}}_{k}=[x_k,y_k]^T \in \mathbb{R}^{2\times 1}$, $k\in \mathcal{K}$. The UAV is assumed to fly at a fixed altitude $H$ above ground and the time-varying horizontal coordinate of the UAV  at time instant $t$ is denoted by $\mathbf{q}(t)=[x(t), y(t)]^T\in \mathbb{R}^{2\times 1}$, with $0\leq t\leq T$.
To serve ground users periodically,  we assume that the UAV needs to return to its initial location by the end of each period $T$, i.e.,
$\q(0) = \q(T)$. 
In addition, the UAV trajectory is also subject to the maximum speed constraints in practice, i.e,
$||\dot{\q}(t)|| \leq V_{\max},  0\leq t \leq T$, 
where $V_{\max}$ denotes the maximum UAV speed in meter/second (m/s).  However, the continuous variable $t$ essentially implies an infinite number of UAV speed constraints which are difficult to tackle in general. To facilitate our analysis and algorithm design, we apply the discrete linear state-space approximation technique, which results in a finite number of constraints. Specifically, we discretize the UAV flight period $T$ into $N$ equally-spaced time slots with step size $\delta_t$, i.e., $t=n\delta_t, n\in \N=\{1,...,N\}$. {Note that for the given maximum UAV speed $V_{\max}$ and altitude $H$, the number of time slots $N$ can be chosen sufficiently large such that the UAV's location change within each time slot $\delta_t$ can be assumed to be negligible, compared to the link distances from the UAV to all ground users. However,  a larger value of $N$ inevitably increases the complexity of the proposed design. Thus,  the number of time slots $N$ can be properly chosen in practice to balance between the accuracy and algorithm complexity. More discussions on the choice of $N$ as well as its impact on the system performance can be found in \cite{JR:wu2017joint}. }
Based on such a discretization, the UAV trajectory can be characterized by a sequence of UAV locations $\mathbf{q}[n]=[x[n], y[n]]^T$, $n=1,\cdots, N$. As a result, the above constraints can be equivalently modeled as
\begin{align}
\q[1] &= \q[N],\\
||\mathbf{q}[n+1]-\mathbf{q}[n]||^2 &\leq S_{\max}^2,  n=1,...,N-1,
\end{align}
where $S_{\max} \triangleq V_{\max}\delta_t$ is the maximum horizontal distance that the UAV can travel within one time slot. Furthermore, the distance  from the UAV to user $k$ in time slot $n$ is assumed to be a constant that  can be expressed as
 \begin{align}
 d_{k}[n] =  \sqrt{H^2 +||\q[n]-\w_k||^2}.
 \end{align}

The measurement results in \cite{lin2017sky,van2016lte,matolak2017air} have shown that the air-to-ground communication channels are mainly dominated by the LoS links even when the UAV is at a moderate altitude. For example, for the UAV at an altitude of $H= 120$ m, the LoS probability of air-to-ground links in rural environment exceeds 95\% for a horizontal ground distance of 4 kilometers (Km) \cite{lin2017sky}. Thus, we assume that the channel quality depends mainly on the UAV-user distance for simplicity. In addition, the Doppler effect induced by the UAV mobility  is assumed to be perfectly compensated at the receivers.  Following the free-space path loss model, the channel power gain from the UAV to user $k$ in time slot $n$ can thus be expressed as
   \begin{align}
h_{k}[n]& = \rho_0d^{-2}_{k}[n] =\frac{\rho_0}{ H^2 +||\q[n]-\w_k||^2},
 \end{align}
 where $\rho_0$ denotes the channel power gain at the reference distance $d_0=1$ m.
Denote the total available system bandwidth by $B$ in Hertz (Hz).  The fraction of bandwidth assigned to user $k$ in time slot $n$  is denoted by $\alpha_k[n]$.
In a practical OFDMA system, $\alpha_k[n]$  is in general a discrete value between $0$ and $1$, which increases linearly with the number of subcarriers assigned to user $k$ in time slot $n$.
It is known that when the number of subcarriers is sufficiently large, $\alpha_k[n]$ can be approximated to a continuous value between $0$ and $1$.  Thus, the bandwidth allocation constraint can be expressed as
 \begin{align}
 \sum_{k=1}^{K}\alpha_{k}[n]&\leq 1, \forall\, n,  \label{eq12} \\
0\leq\alpha_{k}[n]&\leq 1, \forall\, k, n.   \label{eq13}
 \end{align}
Denote the transmit power allocated to user $k$ in time slot $n$ by $p_{k}[n]\geq 0$. Then the total transmit power constraint of the UAV can be expressed as
\begin{align}
\sum_{k=1}^{K}p_{k}[n]\leq P_{\max}, \forall\,n,
\end{align}
where $P_{\max}$ is the maximum allowed transmit power of the UAV in each time slot.
Accordingly, the  instantaneous  achievable rate of user $k$ in time slot $n$, denoted by $r_{k}[n]$ in bits/second/Hz (bps/Hz), can be expressed as
\begin{align}
 r_k[n]& =\alpha_k [n] \log_2 \left( 1 +\frac{p_k[n]h_k[n]}{\alpha_k[n]BN_0} \right) \nonumber \\
 &=\alpha_k [n] \log_2 \left( 1 +\frac{p_k[n]\gamma_0}{\alpha_k[n](H^2+||\mathbf{q}[n]-\mathbf{w}_k||^2)} \right),
\end{align}
 where $\gamma_0 \triangleq \frac{\rho_0}{BN_0}$, with $N_0$ denoting the power spectral density of the additive white Gaussian noise (AWGN) at the receivers.
As a result, the average  achievable throughput of user $k$ over $N$ time slots  in bps/Hz, is given by
   \begin{align}
R_k&= \frac{1}{N}\sum_{n=1}^{N} r_k[n].
 \end{align}

 Motivated by the diversified user applications and heterogeneous delay requirements in the forthcoming 5G wireless networks,  we consider both delay-constrained and delay-tolerant services for users in the system. Specifically, a parameter $\theta_k$,  $\theta_k\in [0\,\,1]$,  is introduced to denote the MRR of user $k$, which means that at any time slot, $\theta_k$ fraction of its average throughput over $N$ slots is delay-constrained and the remaining $1-\theta_k$ fraction  is delay-tolerant.
 In particular, $\theta_k=0$ and $\theta_k=1$  indicate that all the services of user $k$ are delay-tolerant and delay-constrained, respectively. As such, the MRR constraint of user $k$ in time slot $n$ can be expressed as
 \begin{align}
  r_k[n]\geq \theta_kR_k, \forall\, k, n,
\end{align}
which implies that in any of the $N$ time slots, at least $\theta_k$ fraction of the average throughput needs to be satisfied for each user $k$. Therefore, a system where some users' services  are all delay-constrained while those of the others are all delay-tolerant, i.e., $\theta_k\in\{0,1\}, \forall\, k$, is a special case of our general  setup.
\subsection{Problem Formulation}
Let $\mathbf{A}=\{\alpha_{k}[n],  k\in \K,n\in\N \}$, $\mathbf{P}=\{p_{k}[n],  k\in \K,n\in\N\}$, and $\mathbf{Q}=\{\mathbf{q}[n], n\in\N\}$.
By taking into account the MRR constraints of all users, we aim to maximize the minimum average throughput among them via jointly optimizing the bandwidth and power allocation  (i.e., $\A$ and $\pow$) as well as the UAV trajectory (i.e., $\Q$). Define $\eta \triangleq \min \limits_{k \in \mathcal{K}}~ R_k$. The optimization problem is formulated as
 \begin{subequations}\label{probm6}
 \begin{align}
&\max  \limits_{\eta,\mathbf{A},\mathbf{P}, \mathbf{Q}} ~~~~ ~\eta  \\ 
&~~~\text{s.t.}  ~~~~ R_k \geq \eta, \forall\, k, \label{eq601} \\
&~~~~~~~~~ ~ r_k[n]\geq \theta_kR_k, \forall\, k, n, \label{eq602} \\
&~~~~~~~~~ ~  \sum_{k=1}^{K}p_{k}[n]\leq P_{\max}, \forall\, n,  \label{eq603} \\
&~~~~~~~~~ ~  p_{k}[n]\geq 0, \forall\, k, n,  \label{eq6033} \\
&~~~~~~~~~ ~  \sum_{k=1}^{K}\alpha_{k}[n]\leq 1, \forall\, n,  \label{eq604} \\
&~~~~~~~~~ ~ 0\leq\alpha_{k}[n]\leq 1, \forall\, k, n,   \label{eq605} \\
&~~~~~~~~~ ~||\mathbf{q}[n+1]-\mathbf{q}[n]||^2 \leq S_{\max}^2,  n=1,...,N-1,  \label{eq606}\\
&~~~~~~~~~ ~ \mathbf{q}[1]=\mathbf{q}[N]. \label{eq607}
 \end{align}
 \end{subequations}
 Note that the challenges of solving problem (\ref{probm6}) lie in the following three aspects. First, $R_k$ and $r_k[n]$ in constraints (\ref{eq601}) and  (\ref{eq602}) are not jointly concave with respect to the optimization variables $\A$, $\pow$, and $\Q$. Second, for fixed UAV trajectory $\Q$, although $R_k$  and $r_k[n]$ are jointly concave with respect to  $\A$ and $\pow$,  (\ref{eq602}) is non-convex due to the presence of $R_k$ in its left-hand-side (LHS).
Third, for fixed bandwidth and power allocation  $\A$ and $\pow$, $R_k$ and $r_k[n]$ are neither convex nor concave with respect to $\Q$. Consequently, problem (\ref{probm6}) is a non-convex optimization problem and in general, there is no standard method for solving such a problem efficiently. To tackle the above challenges, we first transform  problem (\ref{probm6})  into a more tractable form as follows,
  \begin{subequations}\label{probm66}
 \begin{align}
&\max  \limits_{\eta, \mathbf{A},\mathbf{P}, \mathbf{Q}} ~~~~ ~\eta  \\ 
&~~~\text{s.t.}  ~~~~ r_k[n]\geq \theta_k\eta, \forall\, k, n, \label{eq012} \\
&~~~~~~~~~ ~ \eqref{eq601}, \eqref{eq603}, \eqref{eq6033}, \eqref{eq604}, \eqref{eq605}, \eqref{eq606}, \eqref{eq607}.
 \end{align}
 \end{subequations}
 {Comparing (\ref{eq012}) with (\ref{eq602}), it follows that the feasible set of problem  (\ref{probm6}) is a subset of that of problem  (\ref{probm66}) in general. However, the equivalence of problems  (\ref{probm6}) and  (\ref{probm66}) holds if all users achieve the equal average throughput in the optimal solution to problem  (\ref{probm6}). This can be easily verified since otherwise the objective value of  (\ref{probm6})  can be further improved by allocating more transmit power and/or bandwidth  to the user with a lower average throughput without violating the total transmit power and bandwidth allocation constraints \eqref{eq603} and \eqref{eq604}.  }
 With such a transformation, we only need to focus on solving problem (\ref{probm66}) in the rest of the paper. Although it is still a non-convex optimization problem, problem (\ref{probm66})  facilitates the development of an efficient algorithm. 
Before proceeding to solve problem (\ref{probm66}), we first show the effect of the users' MRRs on the maximum objective  value of problem (\ref{probm66}).
\begin{theorem}\label{theorem1}
The maximum objective value of problem (\ref{probm66}) is an element-wise non-increasing function with respect to $\{\theta_k\}$.
\end{theorem}
\begin{proof}
{Denote the optimal solutions of problem (\ref{probm66}) with ${\bm{\theta}}^*= \{\theta^*_k,  k\in \K\}$ and ${\hat{\bm{\theta}}}= \{\hat\theta_k, k\in \K\}$ by $S^*=\{\eta^{*}, \alpha^*_k[n],p^*_k[n], \q^*[n],  k\in \K,n\in\N \}$ and  $\hat{S}=\{\hat{\eta}, \hat{\alpha}_k[n],\hat{p}_k[n], \hat{\q}[n], k\in \K,n\in\N \}$, respectively. 
To prove Theorem \ref{theorem1}, we only need to show that  $\eta^*\leq \hat{\eta}$ holds when ${\hat{\bm{\theta}}} \preceq {\bm{\theta}}^*$, where $\preceq$ indicates element-wisely less than or equal to, i.e.,  $\hat{\theta}_k\leq \theta^*_k$, $\forall\, k$. Note that in problem (\ref{probm66}), the MRRs are only involved in constraint (\ref{eq012}). Thus, we have the following inequalities
\begin{align}
r^*_k[n]=\alpha^*_k [n] \log_2 \left( 1 +\frac{p^*_k[n]\gamma_0}{\alpha^*_k[n](H^2+||\mathbf{q}^*[n]-\mathbf{w}_k||^2)} \right)  \geq \theta^*_k\eta^*\geq \hat{\theta}_k\eta^*, \forall\, k, n,
\end{align}
which implies that $S^*$ is also a feasible solution of problem  (\ref{probm66}) with $\hat{{\bm{\theta}}}$. Since $\hat{\eta}$ is the maximum objective value of problem  (\ref{probm66}) with $\hat{{\bm{\theta}}}$, it follows that $\eta^{*}\leq \hat{\eta}$, which thus completes the proof of Theorem \ref{theorem1}. }
\end{proof}

Theorem \ref{theorem1} sheds light on the fundamental tradeoff between the max-min average throughput and the user communication delay requirement: as the required MRR  $\theta_k$ increases for any user $k$, the max-min average throughput of the system  decreases in general. This is because imposing more stringent minimum-rate requirements on users fundamentally limit the UAV's mobility to fly closer to any user to achieve a better channel since at the same time it needs to  meet the minimum-rate requirements of other users that are not in its heading direction and thus will have degraded channels with it. As a result,  the degree of freedom for exploiting the UAV's mobility via its trajectory design is restricted by such delay/MRR constraints,  thus leading to lower max-min average throughput.
\section{Proposed Solution}
In this section, we propose an efficient parameter-assisted block coordinate descent algorithm for solving problem (\ref{probm66}). Specifically, we first optimize the bandwidth and power allocation for given UAV trajectory and then optimize the UAV trajectory for given bandwidth and power allocation. These two optimization problems are solved alternately until convergence is achieved.
\subsection{Joint Bandwidth and Power Allocation}
Besides being a subproblem of problem \eqref{probm66}, the bandwidth and power allocation optimization problem for given UAV trajectory may also correspond to a practical scenario when the UAV trajectory is pre-specified due to other tasks such as aerial imaging,  rather than being optimized for achieving best communication performance. Specifically, for any given UAV trajectory $\Q$, the bandwidth and power allocation $\{\A,\pow\}$ can be optimized by solving the following problem,

  \begin{subequations}\label{probm7}
 \begin{align}
&\max  \limits_{\eta, \mathbf{A},\mathbf{P}} ~~~~ ~\eta  \\ 
&~~\text{s.t.}  ~~~\frac{1}{N}\sum_{n=1}^{N}  \alpha_k [n] \log_2 \left( 1 +\frac{p_k[n]g_k[n]}{\alpha_k[n]} \right)  \geq \eta, \forall\, k, \label{eq700} \\
&~~~~~~~ ~\alpha_k [n] \log_2 \left( 1 +\frac{p_k[n]g_k[n]}{\alpha_k[n]} \right) \geq \theta_k\eta, \forall\, k, n, \label{eq701} \\
&~~~~~~~ ~  \sum_{k=1}^{K}p_{k}[n]\leq P_{\max}, \forall\, n,  \label{eq702} \\
&~~~~~~~~~ ~  p_{k}[n]\geq 0, \forall\, k, n,  \label{eq7022} \\
&~~~~~~~ ~  \sum_{k=1}^{K}\alpha_{k}[n]\leq 1, \forall\, n,  \label{eq703} \\
&~~~~~~~ ~ 0\leq\alpha_{k}[n]\leq 1, \forall\, k, n,   \label{eq7033}
 \end{align}
 \end{subequations}
 where $g_k[n] \triangleq \frac{h_k[n] }{BN_0}$.
Define $\alpha_k [n] \log_2 \left( 1 +\frac{p_k[n]g_k[n]}{\alpha_k[n]} \right)\triangleq 0$ when $\alpha_k [n]=0$, $\forall\, k,n$, such that the LHSs of both \eqref{eq700} and \eqref{eq701} are continuous with respect to $\alpha_k [n]$ over the whole domain $0 \leq \alpha_k [n]\leq 1$. {As such, problem (\ref{probm7}) is a convex optimization problem since  $\alpha_k [n] \log_2 \left( 1 +\frac{p_k[n]g_k[n]}{\alpha_k[n]} \right)$ in \eqref{eq700} and \eqref{eq701} is jointly concave with respect to $\alpha_k [n]$ and $p_k[n]$ and \eqref{eq702}-\eqref{eq7033} are all affine constraints. Furthermore, it can be verified that Slater's constraint qualification is satisfied for problem (\ref{probm7}) \cite{Boyd}. Therefore, strong duality holds and  the duality gap between problem (\ref{probm7}) and its dual problem is thus zero,
 which means that the optimal solution can be obtained efficiently by applying the Lagrange duality.} The partial Lagrange function of  problem (\ref{probm7}) can be expressed as
{\small \begin{align}\label{eq:lagr}
&~\mathcal{L}(\eta,\A,\pow, \La, \Mu, \Be, \Nu)  \nonumber\\
=& \eta + \sum_{k=1}^{K}\lambda_k\left(  \frac{1}{N}\sum_{n=1}^{N}  \alpha_k [n] \log_2 \left( 1 +\frac{p_k[n]g_k[n]}{\alpha_k[n]} \right)  - \eta \right) +\sum_{n=1}^{N}\beta_n \left(P_{\max} - \sum_{k=1}^{K}p_{k}[n]\right) \nonumber\\
&+\sum_{k=1}^{K}\sum_{n=1}^{N} \mu_{k,n}\left( \alpha_k [n] \log_2 \left( 1 +\frac{p_k[n]g_k[n]}{\alpha_k[n]} \right)  - \theta_k\eta \right)  + \sum_{n=1}^{N}\nu_n \left(1 -\sum_{k=1}^{K}\alpha_{k}[n]\right)\nonumber\\
=&\left(1-\sum_{k=1}^{K}\lambda_k-\sum_{k=1}^{K}\sum_{n=1}^{N} \mu_{k,n}\theta_k\right)\eta +\sum_{k=1}^{K} \sum_{n=1}^{N}\left(\frac{\lambda_k}{N}+\mu_{k,n}\right)  \alpha_{k}[n] \log_2 \left( 1 +\frac{p_k[n]g_k[n]}{\alpha_k[n]} \right)  \nonumber\\
&-\sum_{k=1}^{K} \sum_{n=1}^{N}\beta_{n}p_{k}[n] - \sum_{k=1}^{K} \sum_{n=1}^{N}\nu_{n}\alpha_{k}[n] + \sum_{n=1}^{N}\beta_{n}P_{\max} + \sum_{n=1}^{N}\nu_{n},
\end{align}}
where  $\bm{\lambda}= \{\lambda_{k},\forall\,k\}$, $\bm{\mu}=\{\mu_{k,n},\forall\,k,n\}$, $\bm{\beta}=\{\beta_{n},\forall\,n\}$,  and  $\bm{\nu}=\{\nu_{n},\forall\,n\}$ are the non-negative Lagrange multipliers associated with  constraints \eqref{eq700}, \eqref{eq701}, \eqref{eq702},  and \eqref{eq703}, respectively.  The boundary constraints  \eqref{eq7022} and  \eqref{eq7033} will be absorbed into the optimal solution in the following.
Accordingly, the dual function is given by
 \begin{subequations}\label{prob:dualf}
\begin{align}
f( \bm{\lambda}, \bm{\mu}, \bm{\beta}, \bm{\nu}) =  &\max \limits_{\eta, \A,\pow} ~\mathcal{L}(\eta,\A,\pow, \La, \Mu, \Be, \Nu)  \\
&~\text{s.t.} ~~~  p_{k}[n]\geq 0, \forall\, k, n, \\
& ~~~~~~~      0\leq\alpha_{k}[n]\leq 1, \forall\, k, n,
\end{align}
for which the following lemma holds.
\begin{lemma}\label{lem_multiplier}
To make $f( \bm{\lambda}, \bm{\mu}, \bm{\beta}, \bm{\nu})$ bounded from the above in \eqref{prob:dualf}, i.e., $f( \bm{\lambda}, \bm{\mu}, \bm{\beta}, \bm{\nu})<  + \infty$,  it follows that $\sum_{k=1}^{K}\lambda_k-\sum_{k=1}^{K}\sum_{n=1}^{N} \mu_{k,n}\theta_k= 1$ must hold.
\end{lemma}
\begin{proof}
This is shown by contradiction. If $\sum_{k=1}^{K}\lambda_k-\sum_{k=1}^{K}\sum_{n=1}^{N} \mu_{k,n}\theta_k>1$ or  $\sum_{k=1}^{K}\lambda_k-\sum_{k=1}^{K}\sum_{n=1}^{N} \mu_{k,n}\theta_k<1$, it follows that $f( \bm{\lambda}, \bm{\mu}, \bm{\beta}, \bm{\nu}) \rightarrow + \infty$ by setting $\eta \rightarrow - \infty$ or $ \eta \rightarrow + \infty$. Thus, neither of the above two inequalities can be true and the lemma is proved.
\end{proof}
 \end{subequations}
 Lemma \ref{lem_multiplier} imposes additional constraints for dual variables  $\bm{\lambda}$ and $\bm{\mu}$. As such, the  dual problem of problem (\ref{probm7}) is given by
 \begin{subequations}\label{prob:dual}
\begin{align}
 &\min \limits_{\La, \Mu, \Be, \Nu} ~f( \La, \Mu, \Be, \Nu)  \\
&~~\text{s.t.} ~~ \sum_{k=1}^{K}\lambda_k-\sum_{k=1}^{K}\sum_{n=1}^{N} \mu_{k,n}\theta_k= 1, \label{eq:13b}\\
&~~~ ~~~~~  \La \succeq 0, \Mu \succeq 0, \Be \succeq 0, \Nu \succeq 0.
\end{align}
 \end{subequations}
Next, we show how to obtain the primal optimal solution by applying the Lagrange duality.

\subsubsection{Obtaining $f( \La, \Mu, \Be, \Nu)$ by Solving Problem \eqref{prob:dualf} }
With the given dual variables, problem \eqref{prob:dualf} can be decomposed into $KN+1$ subproblems that can be solved independently in parallel. Specifically,  one subproblem is  for optimizing $\eta$ and the other $KN$ subproblems are for optimizing $\A$ and $\pow$, i.e.,
\begin{align}\label{subprob:eta}
\max \limits_{\eta} ~\left(1-\sum_{k=1}^{K}\lambda_k-\sum_{k=1}^{K}\sum_{n=1}^{N} \mu_{k,n}\theta_k\right)\eta,
\end{align}
\begin{subequations}\label{subprob:bp}
\begin{align}
&\max \limits_{ \alpha_{k}[n], p_k[n]} ~\left(\frac{\lambda_k}{N}+\mu_{k,n}\right)  \alpha_{k}[n] \log_2 \left( 1 +\frac{p_k[n]g_k[n]}{\alpha_k[n]} \right) - \beta_{n}p_{k}[n] -\nu_{n}\alpha_{k}[n] \\
&~~~~~\text{s.t.} ~~~~  p_{k}[n]\geq 0, \forall\, k, n, \\
& ~~~~~~~~~~~~      0\leq\alpha_{k}[n]\leq 1, \forall\, k, n,
\end{align}
\end{subequations}
where each subproblem in \eqref{subprob:bp} is for user $k$ in time slot $n$. For problem \eqref{subprob:eta}, since $\sum_{k=1}^{K}\lambda_k-\sum_{k=1}^{K}\sum_{n=1}^{N} \mu_{k,n}\theta_k= 1$ holds as in Lemma \ref{lem_multiplier}, the objective value is zero which is independent of the value of $\eta$. This implies that we can choose any arbitrary real number as the optimal solution, denoted by $\eta^{\star}$.  Without loss of generality, we simply set $\eta^{\star}=0$ for the purpose of obtaining the dual function $f( \La, \Mu, \Be, \Nu)$ and updating the dual variables\footnote{We note that $\eta^{\star}=0$ cannot be the optimal primal solution to problem \eqref{probm7}. How to obtain the optimal primal solution for this problem will be discussed later in Section III-A-3). }. Next, we consider problem \eqref{subprob:bp}. Since problem \eqref{subprob:bp} is jointly concave with respect to  $p_{k}[n]$ and $\alpha_k[n]$, the solution that satisfies the Karush-Kuhn-Tucker (KKT) conditions is also the optimal solution.  By taking the derivative of the objective function of \eqref{subprob:bp} with respect to $p_{k}[n]$, the optimal power allocation, denoted by $p^{\star}_{k}[n]$, can be obtained as
\begin{align}\label{powallocation}
p^{\star}_{k}[n]= \alpha_{k}[n] \left[ \frac{\lambda_k +N\mu_{k,n}}{N\beta_n\ln2} - \frac{1 }{g_{k}[n]}  \right]^{+},
\end{align}
where  $[x]^+\triangleq \max\{x,0\}$. Let $\widetilde{p}_{k}[n] \triangleq  \frac{p^{\star}_{k}[n]}{\alpha_{k}[n]}= \left[ \frac{\lambda_k +N\mu_{k,n}}{N\beta_n\ln2} - \frac{1}{g_{k}[n]}  \right]^{+}$, which can be regarded as the power spectrum density of user $k$ in time slot $n$.
Note that in \eqref{powallocation}, the power allocation follows a multi-level water-filling structure. Substituting the obtained $p^{\star}_{k}[n]$ into problem \eqref{subprob:bp} yields
\begin{subequations}\label{subprob:bp2}
\begin{align}
&\max \limits_{ \alpha_{k}[n] } ~\left(\frac{\lambda_k+N\mu_{k,n}}{N}\log_2 \left( 1 +\widetilde{p}_k[n]g_k[n] \right) - \beta_{n}\widetilde{p}_{k}[n] -\nu_{n}\right)\alpha_{k}[n]\\
&~~\text{s.t.} ~~~ 0\leq\alpha_{k}[n]\leq 1, \forall\, k, n.
\end{align}
\end{subequations}
It is evident that problem \eqref{subprob:bp2} is a linear program (LP) with only one optimization variable, $\alpha_{k}[n]$. Thus, the optimal bandwidth allocation, denoted by $\alpha^{\star}_{k}[n]$, can be obtained as
\begin{align}\label{eq:bandwidth}
\alpha^{\star}_k[n]=\left\{\begin{aligned}
                       1,\qquad & \text{if} \left(\frac{\lambda_k+N\mu_{k,n}}{N}\log_2 \left( 1 +\widetilde{p}_k[n]g_k[n] \right) - \beta_{n}\widetilde{p}_{k}[n] -\nu_{n}\right) > 0, \\
                       {a},\qquad & \text{if} \left(\frac{\lambda_k+N\mu_{k,n}}{N}\log_2 \left( 1 +\widetilde{p}_k[n]g_k[n] \right) - \beta_{n}\widetilde{p}_{k}[n] -\nu_{n}\right) = 0, \\
                       0,\qquad & \text{otherwise}, \forall\, k,n.
                    \end{aligned}
             \right.
\end{align}
where $a$ can be any arbitrary real number between 0 and 1 since the objective value of problem \eqref{subprob:bp2} is not affected in this case. For simplicity, we set $a=0$ as for the case of $\eta$. In general,  \eqref{eq:bandwidth} cannot provide the optimal primal solution for problem \eqref{probm7} even with optimal dual variables. Nevertheless, with the above proposed solutions to problems \eqref{subprob:eta} and \eqref{subprob:bp}, the dual function $f( \La, \Mu, \Be, \Nu)$ is  obtained.

\subsubsection{Obtaining Optimal Dual Solution to Problem \eqref{prob:dual} }
After obtaining  $(\eta^{\star}, \A^{\star}, \pow^{\star})$ for given $\bm{\lambda}, \bm{\mu}, \bm{\beta}$, and $\bm{\nu}$, we next solve the dual problem \eqref{prob:dual} to find the optimal dual variables that maximize $f( \La, \Mu, \Be, \Nu)$. Note that although the dual function  $f( \La, \Mu, \Be, \Nu)$ is always convex by definition, it is non-differentiable in general.  As a result, the commonly used subgradient based method such as the ellipsoid method,  can be used to solve  problem \eqref{prob:dual}. 
In each iteration, the dual variables $\La, \Mu, \Be,$ and $\Nu$ are updated based on the subgradients of both the objective function and the constraint functions in problem \eqref{prob:dual}. Specifically, the subgradient of the objective function is denoted by
$\bm{s}_0= [ \Delta\La^T, \Delta\Mu^T,  \Delta\Be^T,  \Delta \Nu^T ]^{T} $ where $\Delta\La, \Delta\Mu,  \Delta\Be$, and  $\Delta\Nu$ are  vectors with the elements respectively given by
\begin{align}
\Delta\lambda_{k}& =  \frac{1}{N}\sum_{n=1}^{N}  \alpha_k [n] \log_2 \left( 1 +\frac{p_k[n]g_k[n]}{\alpha_k[n]} \right), \forall\,k, \\
\Delta\mu_{k,n}  &=  \alpha_k [n] \log_2 \left( 1 +\frac{p_k[n]g_k[n]}{\alpha_k[n]} \right),  \forall\,k,n, \\
\Delta\beta_{n} &= P_{\max} - \sum_{k=1}^{K}p_{k}[n], \forall\,n,\\
\Delta\nu_{n} &= 1 -\sum_{k=1}^{K}\alpha_{k}[n], \forall\,n.
\end{align}
Furthermore, the equality constraint in \eqref{eq:13b} is equivalent to two inequality constraints: $1- \sum_{k=1}^{K}\lambda_k-\sum_{k=1}^{K}\sum_{n=1}^{N} \mu_{k,n}\theta_k\leq 0$ and $-1 + \sum_{k=1}^{K}\lambda_k+\sum_{k=1}^{K}\sum_{n=1}^{N} \mu_{k,n}\theta_k\leq 0$. Thus, the subgradient  of the former constraint function is denoted by $\bm{s}_1= [ \Delta\La^T, \Delta\Mu^T,  \Delta\Be^T,  \Delta \Nu^T ]^{T} $ where the corresponding elements are given by  $\Delta \lambda_k = -1, \Delta \mu_{k,n} = -\theta_k, \Delta \beta_{n} = 0, \Delta \nu_{n} = 0$. In addition, the subgradient of the latter constraint function is given by $\bm{s}_2=-\bm{s}_1$.
With the above subgradients, the dual variables can be updated by the constrained ellipsoid method toward the optimal solution with global convergence \cite{ellipsoid}.

  \begin{algorithm}[t]
\caption{ Joint bandwidth and power allocation algorithm for solving problem  (\ref{probm7}).}\label{Algo:bdpow}
\begin{algorithmic}[1]
\STATE Initialize $\La, \Mu, \Be$, $\Nu$, and the ellipsoid.
\REPEAT
\STATE Solve problem \eqref{subprob:eta} and \eqref{subprob:bp} to obtain $\eta^{\star}$, $\A^{\star}$, and $\pow^{\star}$.
\STATE Compute the subgradients of the objective function and the constraint functions in problem \eqref{prob:dual}.
\STATE Update $\La$, $\Mu$, $\Be$, and $\Nu$ by using the constrained ellipsoid method.
\UNTIL $\La$, $\Mu$, $\Be$, and $\Nu$ converge within a prescribed accuracy.
\STATE Set ($\La^*$, $\Mu^*$, $\Be^*$, $\Nu^*$) $\leftarrow$  ($\La$, $\Mu$, $\Be$, $\Nu$).
\STATE Obtain $\eta^*$, $\A^*$, and $\pow^*$ by solving problem \eqref{probm:primary} and  using \eqref{powallocation}.
\end{algorithmic}
\end{algorithm}
\subsubsection{Constructing Optimal Primal Solution to Problem \eqref{probm7} }
Based on the obtained dual optimal solution  $\La^*$, $\Mu^*$, $\Be^*$, and  $\Nu^*$, it remains to obtain the optimal primal solution  $\{\eta^*, \A^*,\pow^*\}$ to problem \eqref{probm7}.
It is worth pointing out that for a convex optimization problem, the optimal solution that maximizes the Lagrangian function is the optimal primal solution if and only if such a solution is unique and primal feasible \cite{Boyd}. However, in our case, the optimal solutions $\eta^{\star}$ and  $\A^{\star}$ that maximize $\mathcal{L}(\eta,\A,\pow, \La, \Mu, \Be, \Nu)$ are not unique given that $\sum_{k=1}^{K}\lambda_k-\sum_{k=1}^{K}\sum_{n=1}^{N} \mu_{k,n}\theta_k= 1$ and $\left(\frac{\lambda_k+N\mu_{k,n}}{N}\log_2 \left( 1 + \widetilde{p}_k[n]g_k[n] \right) - \beta_{n}\widetilde{p}_{k}[n] -\mu_{n}\right) = 0$ as in \eqref{subprob:eta} and \eqref{eq:bandwidth}. Therefore, additional steps are needed in order to construct the optimal primal  solution. The key observation is that with given $\La^*$, $\Mu^*$, and  $\Be^*$, the optimal power spectrum density (the ratio of the optimal power allocation to the optimal bandwidth allocation), i.e., $\widetilde{p}^*_{k}[n] = \frac{p^*_{k}[n]}{\alpha^*_k[n]}$, can be uniquely obtained from \eqref{powallocation}.
 By substituting  $\widetilde{p}^*_{k}[n]$ into the primal problem \eqref{probm7}, we have
  \begin{subequations}\label{probm:primary}
 \begin{align}
&\max  \limits_{\eta, \mathbf{A}} ~~~~ ~\eta  \\ 
&~~\text{s.t.}  ~~~\frac{1}{N}\sum_{n=1}^{N}  \alpha_k [n] \log_2 \left( 1 + \widetilde{p}^*_k[n]g_k[n] \right)  \geq \eta, \forall\, k, \label{eq:2300} \\
&~~~~~~~ ~\alpha_k [n] \log_2 \left( 1 +\widetilde{p}^*_k[n]g_k[n] \right) \geq \theta_k\eta, \forall\, k, n, \label{eq:2301} \\
&~~~~~~~ ~  \sum_{k=1}^{K} \alpha_k [n]\widetilde{p}^*_{k}[n]\leq P_{\max}, \forall\, n,  \label{eq:2302} \\
&~~~~~~~ ~  \sum_{k=1}^{K}\alpha_{k}[n]\leq 1, \forall\, n,  \label{eq2303} \\
&~~~~~~~ ~ 0\leq\alpha_{k}[n]\leq 1, \forall\, k, n.   \label{eq23033}
 \end{align}
  \end{subequations}
Given $\widetilde{p}^*_{k}[n]$, it is easy to observe that problem \eqref{probm:primary} is an LP with respect to  $\A$ and $\eta$, which thus can be efficiently solved by using standard convex optimization solvers such as CVX \cite{cvx}. After obtaining the optimal bandwidth allocation $\A^*$, the corresponding power allocation $\pow^*$ can be obtained as  $p^*_{k}[n] = \widetilde{p}^*_{k}[n] \alpha^*_k[n]$, $\forall\, k,n$.  The details of the procedures for obtaining the optimal solution to problem \eqref{probm7}  are summarized in Algorithm \ref{Algo:bdpow}. { The computational complexity of Algorithm \ref{Algo:bdpow} consists of three parts. The first part is for solving problems \eqref{subprob:eta} and \eqref{subprob:bp}, the second part is for updating the dual variables by the ellipsoid method, and the third part is for solving linear program problem  \eqref{probm:primary}. In step 3) of Algorithm 1, the complexity of solving problem  \eqref{subprob:eta} is $O(1)$ and that of solving \eqref{subprob:bp} is $O(KN)$. The complexities of step 4) and 5) are $O(KN)$ and $O(K^2N^2)$, respectively. Since the ellipsoid method takes $O(K^2N^2)$ to converge, the total complexity for step 2) to 6) is $O(K^4N^4)$ \cite{ellipsoid}. The complexity of  solving   \eqref{probm:primary} is $O(K^3N^3)$. Therefore, the total complexity of Algorithm \ref{Algo:bdpow} is $O(K^4N^4)$.  }

\subsection{UAV Trajectory Optimization}
Given any feasible bandwidth and power allocation $\{\A,\pow\}$, problem \eqref{probm66} is simplified into the following problem for  optimizing the UAV trajectory $\Q$ only, i.e.,

\begin{subequations}\label{probm8}
 \begin{align}
&\max  \limits_{\eta, \mathbf{Q}} ~~~~ ~\eta  \\ 
&~~\text{s.t.}  ~~~\frac{1}{N}\sum_{n=1}^{N}  \alpha_k [n] \log_2 \left( 1 +\frac{\gamma_k[n]}{H^2+||\mathbf{q}[n]-\mathbf{w}_k||^2} \right)  \geq \eta, \forall\, k, \label{eq0w12} \\
&~~~~~~~ ~\alpha_k [n] \log_2 \left( 1 +\frac{\gamma_k[n]}{H^2+||\mathbf{q}[n]-\mathbf{w}_k||^2} \right) \geq \theta_k\eta, \forall\, k, n, \label{eqq012} \\
&~~~~~~~ ~||\mathbf{q}[n+1]-\mathbf{q}[n]||^2 \leq S_{\max}^2,  n=1,...,N-1,  \label{eq14}\\
&~~~~~~~ ~ \mathbf{q}[1]=\mathbf{q}[N], \label{eq15}
 \end{align}
 \end{subequations}
 where $\gamma_k[n] \triangleq \frac{p_k[n]\rho_0}{\alpha_k[n]BN_0}$. Note that problem \eqref{probm8} is not a convex optimization problem since the LHSs of constraints \eqref{eq0w12}  and \eqref{eqq012} are not concave with respect to $\mathbf{q}[n]$.  In general, there is no efficient method to obtain the optimal solution for such a non-convex problem. However, we observe that  the LHSs of both \eqref{eq0w12}  and \eqref{eqq012} are convex with respect to $||\mathbf{q}[n]-\mathbf{w}_k||^2$.
 Note that for a convex function, its
first-order Taylor expansion is the global under-estimator at any point \cite{Boyd}. This thus motivates us to leverage the successive convex optimization technique to tackle the non-convex problem \eqref{probm8} by an iterative algorithm, where in each iteration, the LHSs of both \eqref{eq0w12}  and \eqref{eqq012} are replaced by more tractable functions derived from the Taylor expansion at a given local point.  Specifically, with given local point $\mathbf{q}^r[n]$,  we have the following inequality
   \begin{align}\label{eq155}
{r}_k[n]&=  \alpha_k[n]   \log_2\left( 1 +\frac{\gamma_0}{H^2+||\mathbf{q}[n]-\mathbf{w}_k||^2} \right)  \nonumber\\
& \geq \alpha_k[n]\left(-A^r_k[n]\left(||\mathbf{q}[n]-\mathbf{w}_k||^2 -||\mathbf{q}^r[n]-\mathbf{w}_k||^2 \right) + B^r_k[n]\right)  \nonumber\\
&\triangleq {r}^{{\rm lb},r}_k[n],
\end{align}
where 
\begin{align}
A^r_k[n]& = \frac{\gamma_0\log_2e}{(H^2+||\mathbf{q}^r[n]-\mathbf{w}_k||^2)(H^2+||\mathbf{q}^r[n]-\mathbf{w}_k||^2+\gamma_0)}, \\
B^r_k[n]&= \log_2\left(1+\frac{\gamma_0}{H^2+||\mathbf{q}^r[n]-\mathbf{w}_k||^2}\right), \forall\, k,n.
\end{align}

For any given local point $\Q^r$, define the function $\eta^{{\rm lb}, r}(\A,\Q)= \min \limits_{i\in \mathcal{K}}~ \sum_{n=1}^{N}{r}^{{\rm lb},r}_k[n]$.  With the lower bounds ${r}^{{\rm lb}, r}_k[n]$, $\forall\, k$,  in (\ref{eq155}) and $\Q^r$, problem (\ref{probm8}) is approximated as the following problem 
\begin{subequations} \label{probm30}
 \begin{align}
 &\mathop {\text{max} }\limits_{\eta^{{\rm lb}, r},\Q}~~\eta^{{\rm lb}, r}   \nonumber\\ %
&~~~\text{s.t.} ~~ 
\frac{1}{N}\sum_{n=1}^{N} {r}^{{\rm lb},r}_k[n] \geq \eta^{{\rm lb}, r}, \forall\, k,  \label{eq31} \\
&~~~~~~~~  {r}^{{\rm lb},r}_k[n] \geq \theta_k \eta^{{\rm lb}, r}, \forall\, k, n,  \label{eq32} \\
&~~~~~~~~ ||\mathbf{q}[n+1]-\mathbf{q}[n]||^2\leq  S_{\max}^2,  n=1,...,N-1,  \label{eq313} \\
&~~~~~~~~ \mathbf{q}[1]=\mathbf{q}[N]. \label{eq151}
 \end{align}
 \end{subequations}
Note that constraints (\ref{eq31}), (\ref{eq32}),  and (\ref{eq313}) are all convex quadratic constraints and (\ref{eq151}) is a linear constraint. {Therefore, problem (\ref{probm30}) is a convex quadratically constrained quadratic program (QCQP) that
can be solved within a polynomial complexity  by standard convex optimization solvers such as  CVX \cite{cvx}.} It is worth pointing out that due to the lower bounds adopted from \eqref{eq155}, constraints (\ref{eq31}) and  (\ref{eq32}) imply (\ref{eq0w12}) and \eqref{eqq012}, respectively, but the reverse does not hold in general.
In this regard, the optimal objective value obtained by solving problem (\ref{probm30}) always serves as a lower bound for that of problem  (\ref{probm8}). %

 \subsection{Overall Algorithm Design}
 In this section, we propose an efficient iterative algorithm to solve problem (\ref{probm66}) based on the results in previous two sections. Note that  the conventional block coordinate descent method alternately optimizes one block of optimization variables in each iteration while keeping other blocks of optimization variables fixed, until the convergence is achieved \cite{bertsekas1999nonlinear,hong2016unified}. However, directly applying such a block coordinate descent method for our considered problem, i.e., by alternately optimizing one block from the power and bandwidth allocation variables $\{\A,\pow\}$ in problem \eqref{probm7} and the UAV trajectory variables $\Q$ in problem \eqref{probm8}, with the other block fixed,  may fail to update the UAV trajectory effectively. This can be observed from the MRR constraints in problem (\ref{probm8}), i.e.,
 \begin{align}
\alpha_k [n] \log_2 \left( 1 +\frac{\gamma_k[n]}{H^2+||\mathbf{q}[n]-\mathbf{w}_k||^2} \right) \geq \theta_k\eta, \forall\, k, n. \label{eqq0122}
\end{align}
 For given $\alpha_k[n]$ and $p_k[n]$, $\forall\, k, n$,  since problem (\ref{probm8}) aims to increase $\eta$ by optimizing UAV trajectory $\mathbf{q}[n]$,  the right-hand-side (RHS) of (\ref{eqq0122}) is expected to increase after each iteration. As a result, for users that have met constraints  in \eqref{eq701} with equality in the last iteration, the only way to increase the LHS of \eqref{eqq0122}  in the current iteration is by decreasing $||\mathbf{q}[n]-\mathbf{w}_k||^2$, $\forall\, k$. This implies that in each time slot $n$, the UAV's location $\mathbf{q}[n]$ needs to be updated to decrease the distances from the UAV to all these users. Thus, the freedom for optimizing the UAV trajectory is severely limited which can lead to ineffective update of the UAV trajectory in each iteration. 

 To tackle this issue, we propose a new parameter-assisted block coordinate descent method where each parameter is  a temporary MRR for a corresponding user $k$, denoted by $\theta_{{\rm temp}, k}$, which is set larger than the actual MRR target $\theta_k$, if $0<\theta_k<1$.  The main idea is to use the newly introduced temporary MRR $\theta_{{\rm temp}, k}>\theta_k$ for solving problem (\ref{probm8}) rather than directly using the target  $\theta_k$.
 Specifically, $\theta_{{\rm temp}, k}$ for each user $k$ is gradually decreased before solving problem (\ref{probm8}) in each iteration, until the target MRR $\theta_k$ is achieved for all users. As such,  $\eta$ will increase after each iteration, while the constraints in  (\ref{eqq0122}) will be relaxed due to the decrease of $\theta_{{\rm temp}, k}$'s, which thus permits a more effective  UAV trajectory update in each iteration compared to the conventional block coordinate decent method. Furthermore, the proposed method also generates a feasible solution for the original problem (\ref{probm66}) since $\theta_k$ will be eventually achieved by gradually decreasing $\theta_{{\rm temp}, k}$ with a predefined step size, $\theta_{{\rm step}, k}>0$. The details of the proposed method are summarized in Algorithm \ref{Algo:whole}.

 The convergence of the proposed Algorithm \ref{Algo:whole} can be analyzed as follows.  From step 4, it can be seen that $\theta_k$, $\forall\, k$, is non-increasing  over the iterations. Recall that the objective value of problem  (\ref{probm66}) is element-wise non-increasing with respect to $\theta_k$, $\forall\, k$. Thus, it can be shown that the objective value achieved by Algorithm  \ref{Algo:whole} is non-decreasing in each iteration.  Furthermore,  the objective value of problem  (\ref{probm66}) is bounded from the above. Thus, the proposed Algorithm  \ref{Algo:whole} is guaranteed to converge and the obtained solution is also feasible for the original problem  (\ref{probm6}). {Although the obtained solution is generally suboptimal, we validate the effectiveness of Algorithm 2 via simulations by comparing with other benchmark schemes in Section IV.}

 \subsection{UAV Initial Trajectory and Users' MRR Initialization}
  \begin{algorithm}[t]
\caption{ Parameter-assisted  block coordinate descent algorithm for solving problem (\ref{probm66}).}\label{Algo:whole}
\begin{algorithmic}[1]
\STATE Initialize $\Q^0$, $\theta_{{\rm ini}, k}$, and  $L_{\max}$. Let $r=0$, $\theta_{{\rm temp}, k}=\theta_{{\rm ini}, k}$, and $\theta_{{\rm step}, k}= \frac{\theta_{{\rm ini}, k}- \theta_k}{L_{\max}}$, $\forall\, k$.
\REPEAT
\STATE Solve problem (\ref{probm7}) by applying Algorithm \ref{Algo:bdpow} for given $\Q^r$, and denote the optimal solution as $\{\mathbf{A}^{r+1}, \mathbf{P}^{r+1}\}$.
\STATE $\theta_{{\rm temp}, k}= \max \{ \theta_{{\rm temp},k}-(r+1)\theta_{{\rm step}, k}, \theta_k \}$,  $\forall\, k$.
\STATE Solve problem (\ref{probm8}) for given $\{\mathbf{A}^{r+1}, \mathbf{P}^{r+1}, \Q^{r} , \theta_{{\rm temp},k} \}$, and denote the optimal solution as $\{\mathbf{Q}^{r+1}\}$.
\STATE Update $r=r+1$.
\UNTIL{  $\theta_{{\rm temp},k}= \theta_k$, $\forall\, k$, and the fractional increase of the objective value  is below a threshold $\epsilon>0$.}
\end{algorithmic}
\end{algorithm}

In Algorithm \ref{Algo:whole}, the initial UAV trajectory, $\mathbf{Q}^0=\{\mathbf{q}^0[n], \forall\,n\}$, and the initial temporary users' MRRs, $\theta_{{\rm ini}, k}$, $\forall\, k$, need to be set. First, we propose in the following a general and systematic trajectory initialization scheme based on the simple circular UAV trajectory, which also includes the one proposed in  \cite{wu2017joint} with $\theta_k=0, \forall\, k$,  as a special case. 
For a circular UAV trajectory, we need to obtain the circle center $\cc_{\rm ini}=[x_{\rm ini}, y_{\rm ini}]^T$ and radius  $r_{\rm ini}$. For convenience, the circle center is set as the geometry center of all users, which is given by $\cc_{\rm ini} = \frac{\sum_{k=1}^{K}{\w}_{k}}{K}$. The main idea for determining the circle radius of the initial circular trajectory is based on the following three observations.
First, it is intuitive that as $\theta_k$ increases, the area covered by the UAV trajectory shall decrease, as implied in Theorem 1.
 As such, the initial circle radius $r_{\rm ini}$ should decrease with $\theta_k$, $\forall\, k$.  Second, when $\theta_k=0$, $\forall\, k$,  the initial circle radius should be simplified to that for the case when the MRR constraints are not considered as in  \cite{wu2017joint}. Third, it is also known that when $\theta_k=1$, $\forall\, k$, the UAV mobility does not provide any performance gain, which implies that the UAV should remain static. In such a case, the UAV trajectory becomes a point and the initial circle radius is thus zero. Based on the above observations, we propose to set the initial circle radius as
 \begin{align}\label{eq:circleradius}
 r_{\rm ini}=\left(1-\frac{ \sum_{k=1}^{K}\theta_k}{K}\right)r_0,
 \end{align}
 where $r_0$ is the circle radius for the case when $\theta_k=0, \forall\, k$, or no MRR constraints are considered  \cite{wu2017joint}. According to  \cite{wu2017joint},  $r_0$ is obtained as follows: 1) by using the user geometry center $\cc_{\rm ini} $ as the circle center, we first  calculate the radius of a minimum circle that can cover all ground users, denoted by $r_{\min}=\max\limits_{k\in \mathcal K } ||\mathbf{w}_k-\cc_{\rm ini} ||$; 2)  to balance the number of users outside and inside the UAV trajectory circle as well as satisfying the UAV speed constraint, we set $r_0=\min \{\frac{V_{\max}T}{2\pi},\frac{r_{\min} }{2} \}$. From \eqref{eq:circleradius}, it is easy to check that when $\theta_k=0$, $\forall\, k$, $r_{\rm ini}=r_0$ and when $\theta_k=1$, $\forall\, k$, $r_{\rm ini}=0$, as expected. Let $\varphi_n = 2\pi\frac{n-1}{N-1}$, $\forall\, n$. With $\cc_{\rm ini}$ and  $r_{\rm ini}$, the initial circular UAV trajectory $\Q^0$ is given as follows,
 \begin{align}
\q^0[n] = \left[x_{\rm ini} + r_{\rm ini}\cos\varphi_n,  y_{\rm ini} + r_{\rm ini}\sin\varphi_n\right]^{T}, n=1,...,N.
\end{align}

Next, we initialize the users' MRRs, $\theta_{{\rm ini},k}$, $\forall\,k$. Note that if user $k$ has no MRR constraint, i.e., $\theta_k=0$, then there exists no constraint for user $k$ in \eqref{eqq0122}. Thus, user $k$ will not cause the ineffective UAV trajectory update problem previously described in Section III-C and its initial MRR can be directly set as $\theta_{{\rm ini}, k}=0$. In contrast, as long as user $k$ has a non-trivial MRR target, i.e., $\theta_k>0$,  $N$ corresponding MRR constraints in \eqref{eqq0122}  will be imposed for user $k$ which may cause ineffective updates of the UAV trajectory. Thus, the initial MRRs of users are set as
\begin{equation}\label{eq:initial_ratio}
\theta_{{\rm ini},k}=\left\{\begin{aligned}
                       1,\qquad &  \text{if} ~~ \theta_k >0, \\
                       0, \qquad & \text{if}  ~~\theta_k =0.
                    \end{aligned}
             \right.
\end{equation}



 \section{Numerical Results}
In this section, numerical results are provided to validate the proposed joint OFDMA resource allocation and UAV trajectory design as well as the fundamental throughput-delay tradeoff in UAV-enabled wireless communications. We consider a system with $K=4$ ground users that are located in a horizontal plane as shown in Fig.  \ref{homo1_trj}, marked by `$\times$'s.
{The UAV is assumed to fly at a fixed altitude $H=500$ m \cite{matolak2017air}}.
The total available communication bandwidth is $B = 10$ MHz and the noise power spectrum density at the ground users is assumed to be identical and set as $N_0 = -169$ dBm/Hz. 
 The channel power gain at the reference distance $d_0= 1$ m is set as $\rho_0=-50$ dB. Other parameters are set as $P_{\max}=0.1$ W, $V_{\max}=50$ m/s, $T= 270$ s, and $N=540$, respectively, if not specified otherwise. For illustration, all the trajectories in the simulations are sampled every 4 s and the sampled points are marked by `$\triangle$'s.

%

  \begin{figure}[!t]
\centering
\subfigure[$\theta =0$]{\includegraphics[width=2.5in, height=1.8in]{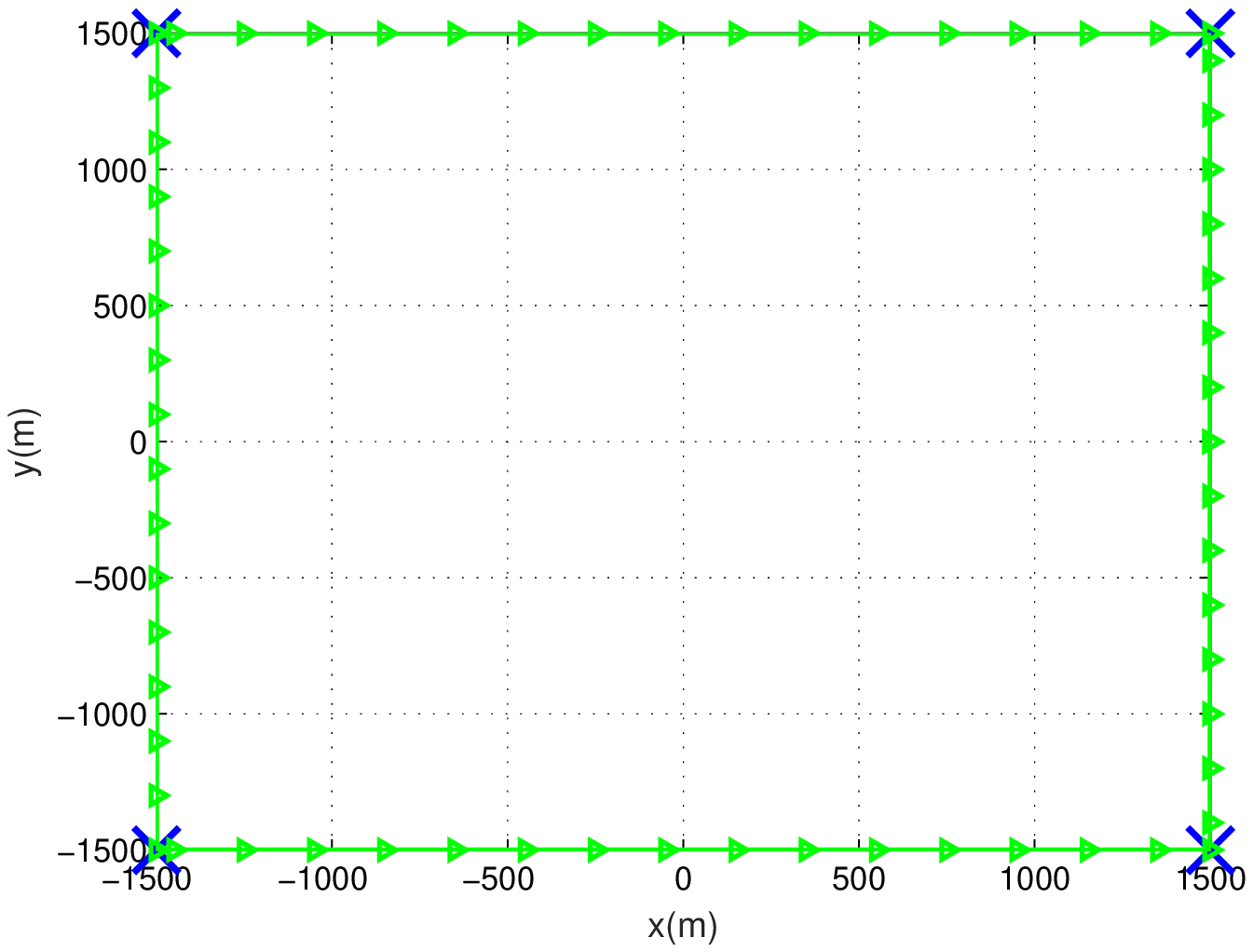}}
\subfigure[$\theta =0.6$]{\includegraphics[width=2.5in, height=1.8in]{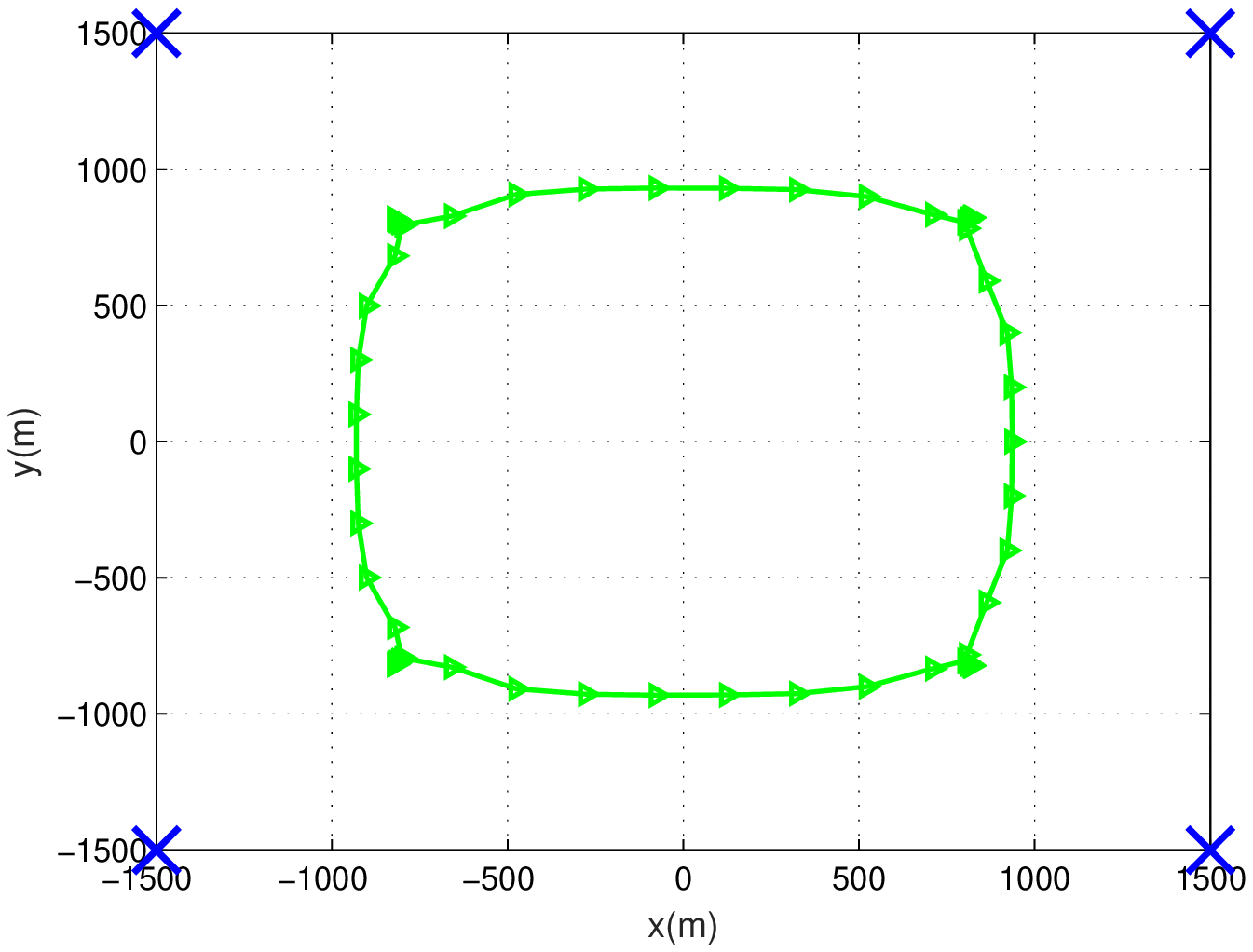}}
\subfigure[$\theta =0.8$]{\includegraphics[width=2.5in, height=1.8in]{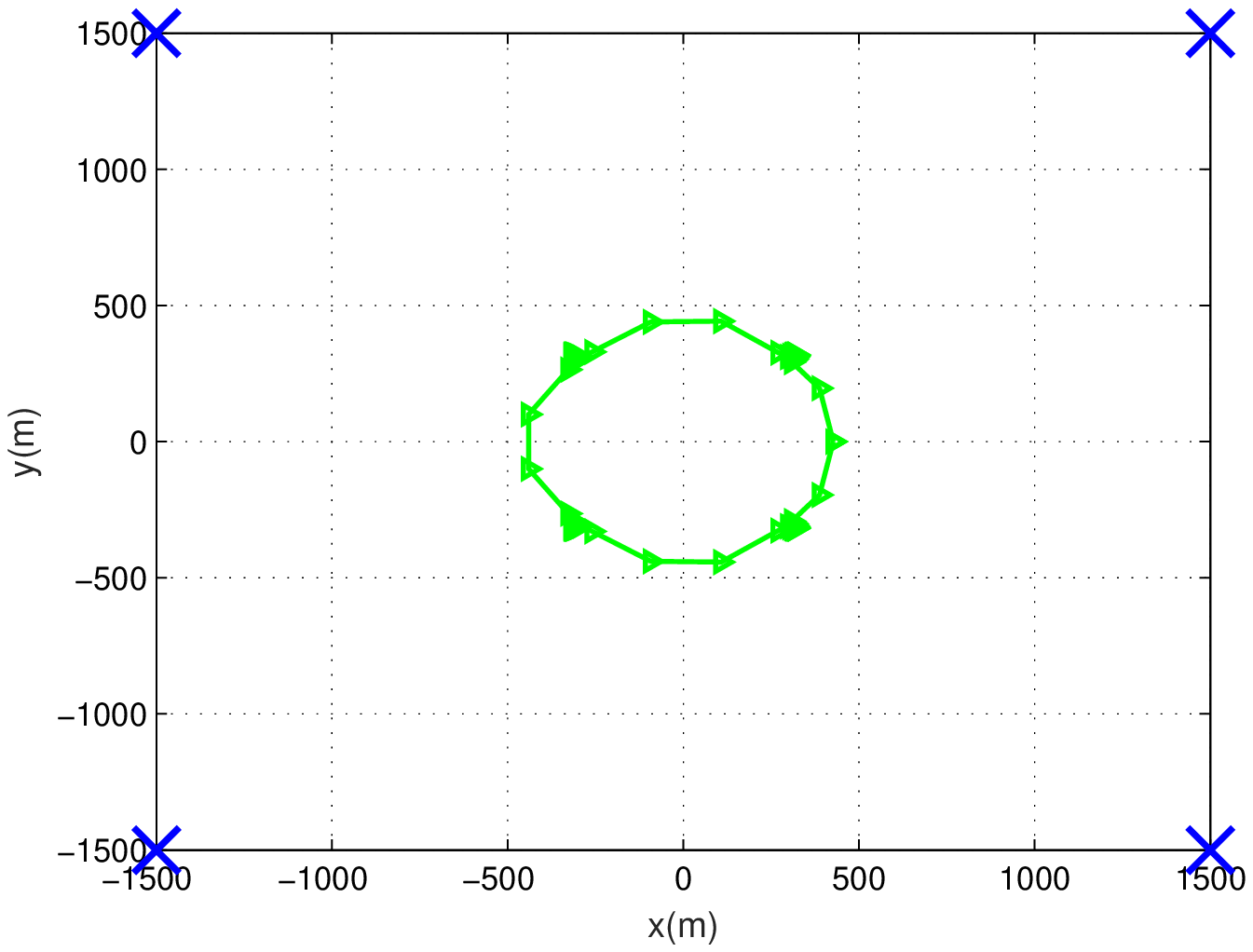}}
\subfigure[$\theta =1$]{\includegraphics[width=2.5in, height=1.8in]{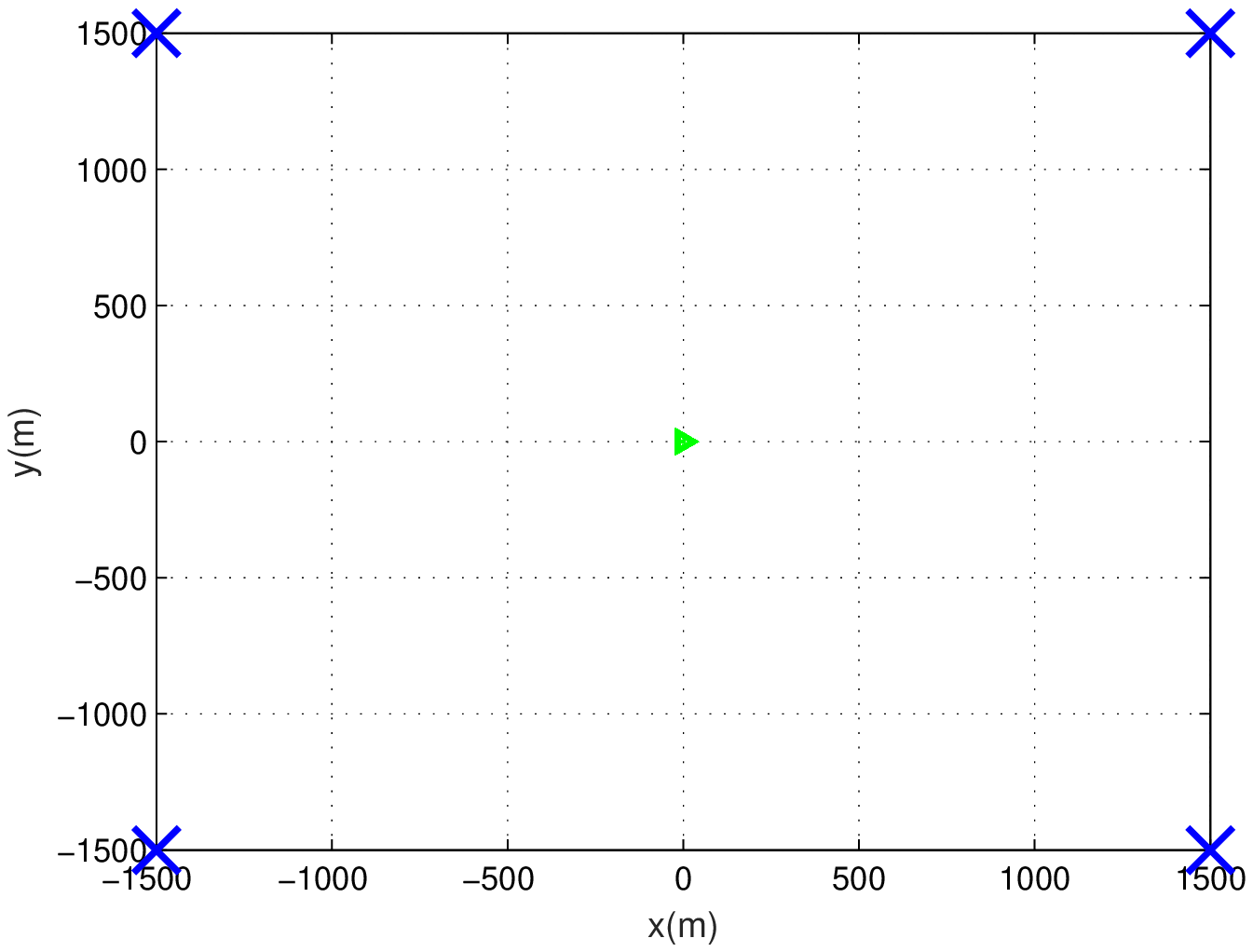}}
\caption{UAV trajectory versus homogeneous MRR  $\theta=\theta_k$, $\forall\, k$, for $T = 270$ s. } \label{homo1_trj}
\end{figure}

\subsection{UAV Trajectory and Max-min Throughput versus  Homogeneous  MRRs}
We first consider the homogeneous delay requirement case when all ground users have the same MRR, i.e., $\theta_k =\theta$, $\forall\, k$.  In Figs. \ref{homo1_trj} and \ref{homo1_rate}, the UAV trajectory and the max-min throughput are illustrated respectively under different MRRs. It can be observed from Fig.  \ref{homo1_trj} that as the MRR, $\theta$, increases, the UAV's flight distance decreases and the UAV's trajectory shrinks gradually from a square to a smaller ellipse and finally a fixed point,  given the same flight period $T$.   In particular, when $\theta=0$, i.e.,  no MRR constraint is considered as in \cite{wu2017joint},  the UAV sequentially visits and stays above each of the ground users by maximally exploiting its mobility. As such,  the best air-to-ground channel can be realized between the UAV and each ground user. However, when $\theta>0$, the UAV's trajectory adaptively changes depending on the value of $\theta$, as expected.
Notice that in this setup,  the closer the UAV flies to one particular user, the farther it is away from some other users inevitably.
As a result, meeting the MRR constraints of these users  will consume more resources (power, bandwidth) and thus becomes the bottleneck for improving the max-min throughput of the system.  Such a situation becomes worse when the MRR and/or the inter-user distance becomes larger.
Generally speaking, with more stringent MRR constraints, the UAV trajectory tends to be more restricted to avoid getting too far away from any of the users.
Finally, as all ground users have the same MRR constraint as well as the same average throughput, the UAV trajectories are shown to be all symmetric over the users in Fig. \ref{homo1_trj}.

    \begin{figure}[!t]
\centering
\includegraphics[width=0.6\textwidth]{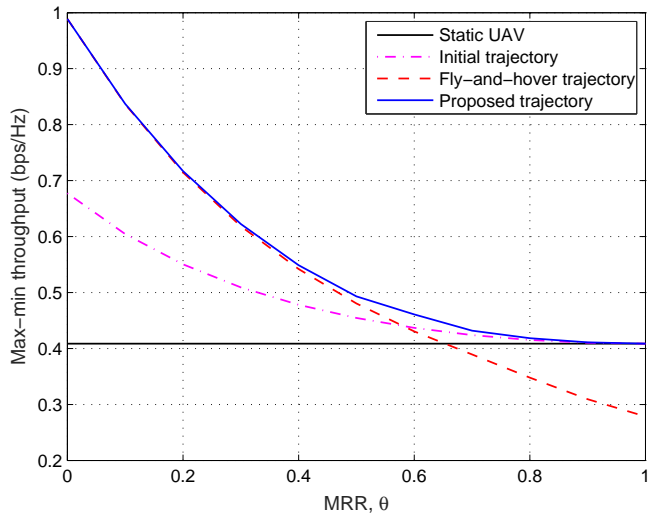}
\caption{Max-min throughput  versus homogeneous MRR $\theta=\theta_k$, $\forall\,k$, for $T=270$ s.} \label{homo1_rate}
\end{figure}

 The effect of the MRR constraint on the max-min average throughput is shown in Fig. \ref{homo1_rate}. Specifically,  we compare the following four trajectories: 1) Proposed trajectory which is obtained by applying Algorithm \ref{Algo:whole}; 2) Fly-and-hover trajectory where the UAV flies with the maximum speed to visit all the users and hovers (with zero speed) above each of them by equally allocating the remaining time\footnote{It is worth pointing out that this trajectory is only feasible when $T$ is sufficient large such that the UAV can visit above each of all ground users.};
 3) Initial circular trajectory as described in Section III-D; and 4) Static UAV where the UAV remains static above the geometry center of the users. For all the four schemes considered, the bandwidth and power allocation is optimized by applying Algorithm \ref{Algo:bdpow} with the corresponding UAV trajectory.
First, we observe  that the max-min average throughput gradually decreases with the MRR for the first three considered  trajectories with a mobile UAV in general. This is due to the fundamental tradeoff between the system throughput and the user delay/MRR constraint, which is in accordance with Theorem \ref{theorem1}. Second, when the MRR is small, the proposed trajectory significantly outperforms the circular trajectory and static UAV. This is because a small MRR in general implies a large degree of freedom for the UAV trajectory design,  but the static UAV cannot exploit the UAV mobility while the circular trajectory does not fully exploit the UAV mobility.
 In contrast, the fly-and-hover trajectory is observed to achieve the throughput very close to the proposed trajectory for small MRR values.  However, when the MRR  becomes large, the fly-and-hover trajectory suffers from a significant throughput loss compared to the proposed trajectory. This is expected since the fly-and-hover trajectory results in highly asymmetric user channels over time and thus is inefficient for meeting increasingly more stringent users' MRR requirements.
 Finally, when the MRR reaches the maximum value of one, the max-min average throughput of the proposed trajectory becomes identical to that of static UAV as the trajectory converges to the same point as the static UAV.

\begin{figure}[]
\centering
\includegraphics[width=0.6\textwidth]{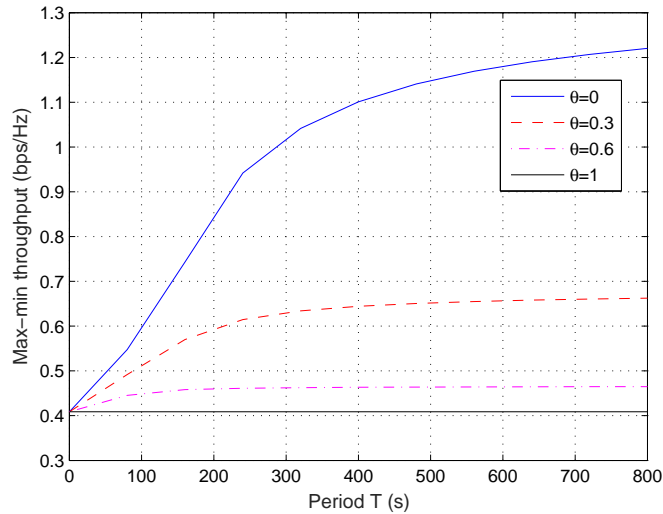}
\caption{Max-min throughput  versus UAV flight period $T$.} \label{homo1:T}
\end{figure}

In Fig. \ref{homo1:T}, the effect of the UAV flight period $T$ on the max-min average throughput is shown under different values of the homogeneous MRR, $\theta$.
It is observed that the max-min throughput in the three cases with $\theta<1$ all increases with $T$, while  for the case of  $\theta = 1$, the max-min throughput  remains constant, regardless of $T$. This suggests that as long as the users have delay-tolerant data traffic, i.e, $\theta<1$,  the UAV mobility indeed provides throughput gains over a static UAV with $\theta=1$ for any $T>0$. In addition, such a throughput gain generally increases with more flight time for the UAV, although the gain is more pronounced when $\theta$ is smaller or less strict delay/MRR constraints are applied.  This is because in such cases, as $T$ increases, the UAV has more time to get closer to each of the users and even be able to hover above them to enjoy the best channels to them. In contrast, if the services required by users are all delay-constrained, i.e., $\theta=1$, then the UAV is unable to achieve any throughput gain over the static case since the optimal UAV trajectory in this case is also a fixed point. The above observations further demonstrate the fundamental throughput-delay tradeoff. 

\begin{figure}[!t]
\centering
\subfigure[$ \theta_3=\theta_4=0$]{\includegraphics[width=2.5in, height=1.8in]{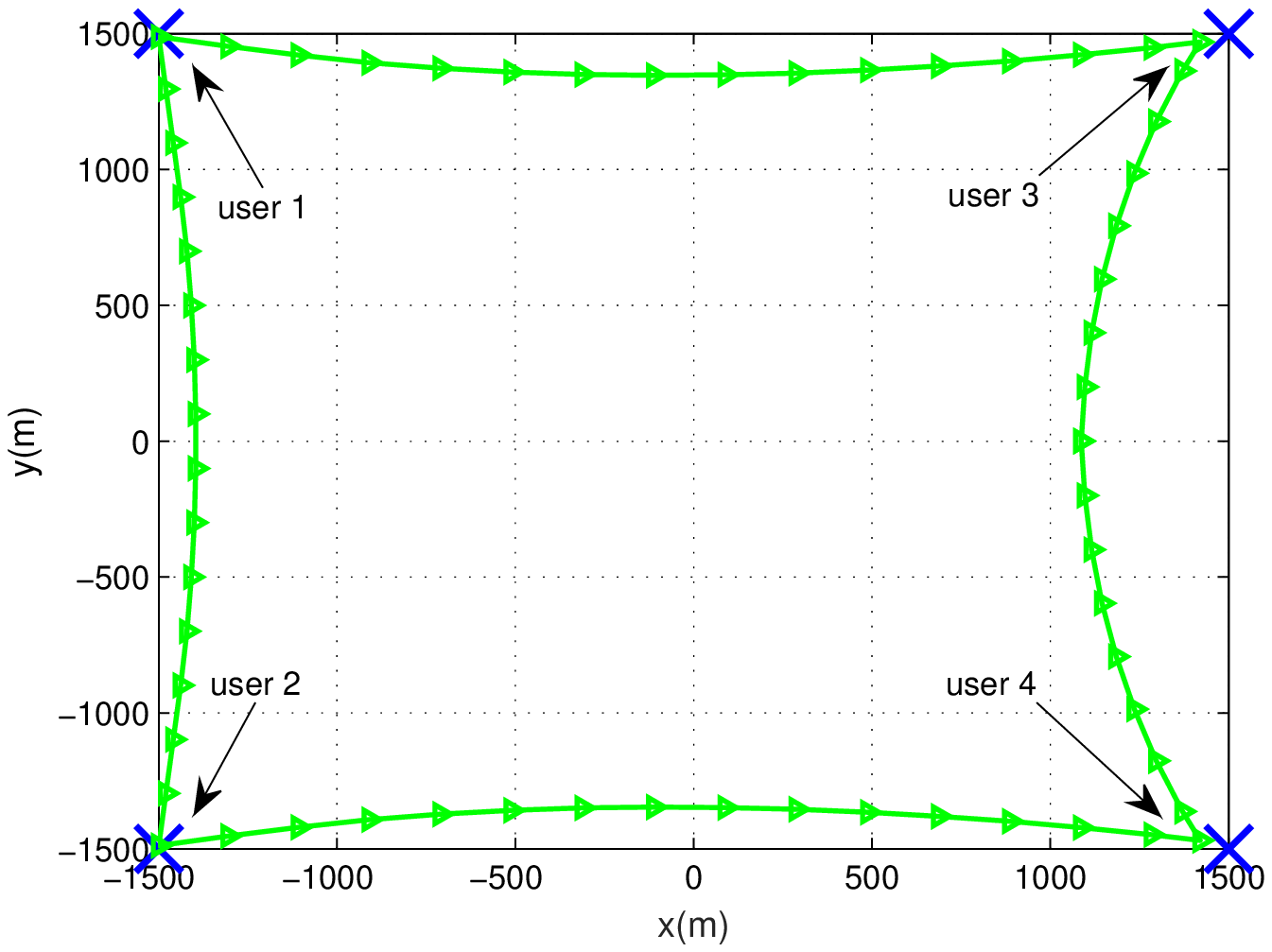}}
\subfigure[$ \theta_3=\theta_4=0.4$]{\includegraphics[width=2.5in, height=1.8in]{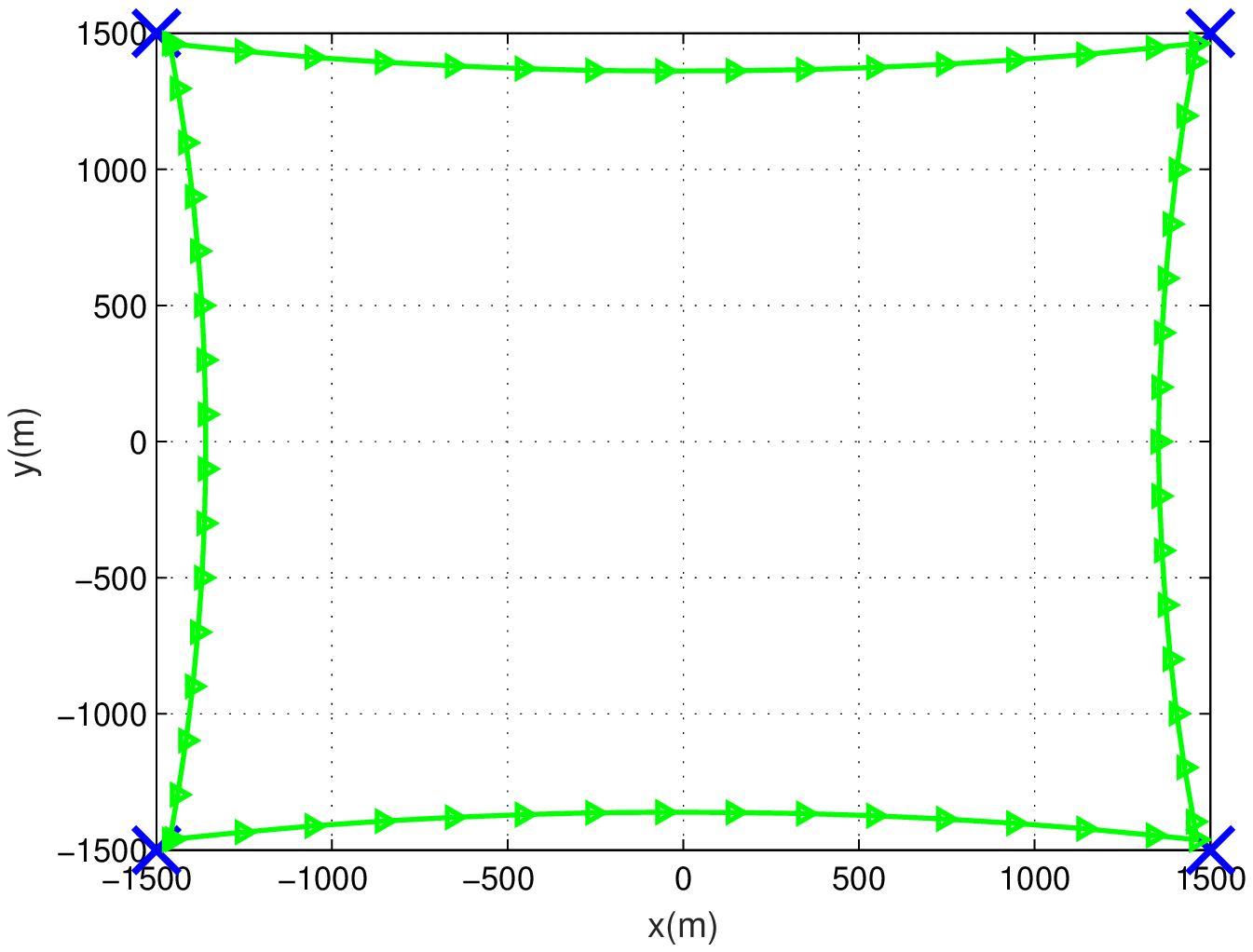}}
\subfigure[$\theta_3=\theta_4=0.6$]{\includegraphics[width=2.5in, height=1.8in]{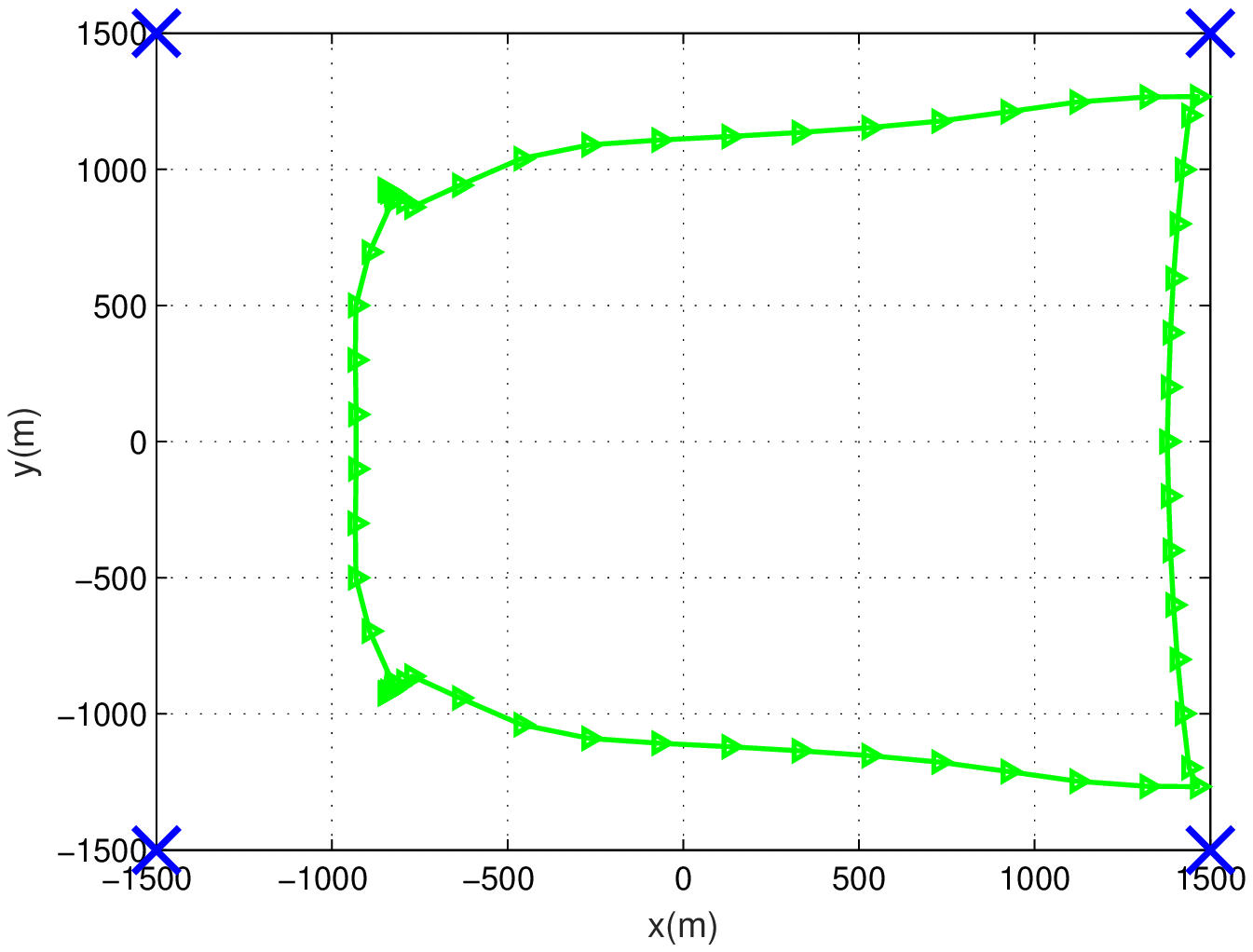}}
\subfigure[$\theta_3=\theta_4=0.8$]{\includegraphics[width=2.5in, height=1.8in]{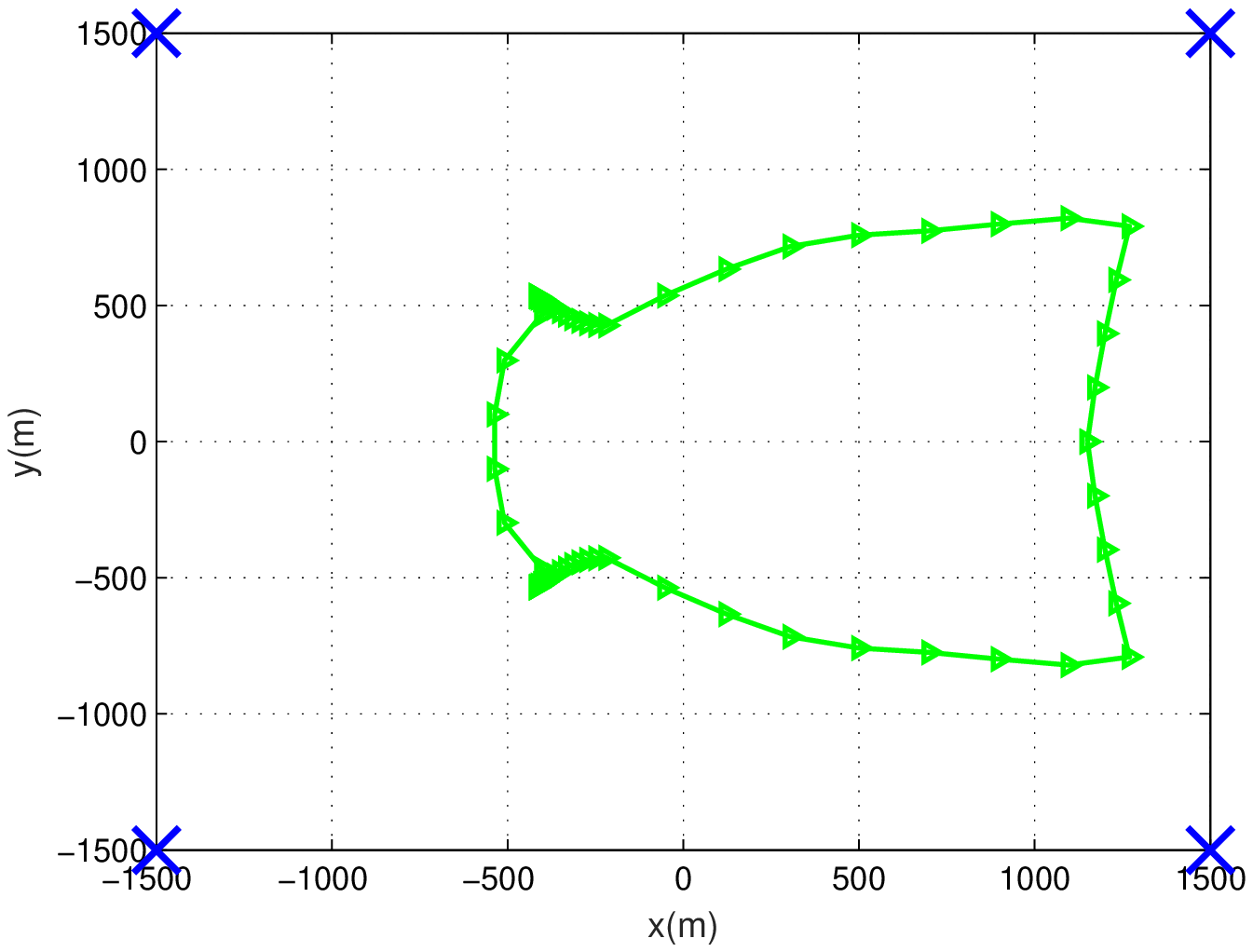}}
\caption{UAV trajectory versus MRRs of equal $\theta_3$ and $\theta_4$ for $T = 270$ s with fixed $\theta_1=\theta_2=0.4$.  } \label{hete1_trj}
\end{figure}

\begin{figure}[!t]
\centering
\includegraphics[width=0.6\textwidth]{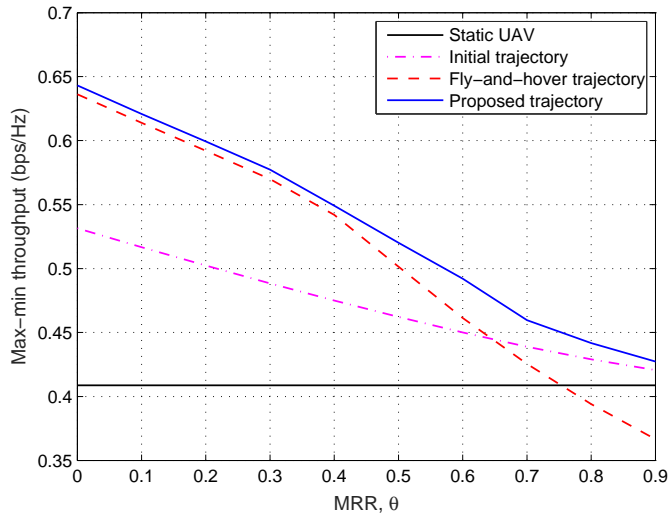}
\caption{Max-min throughput  versus heterogeneous MRRs $\theta=\theta_3=\theta_4$  for $T=270$ s with fixed $\theta_1=\theta_2=0.4$.} \label{hete1_rate}
\end{figure}

\subsection{UAV Trajectory and Max-min Throughput versus  Heterogeneous MRRs }
Next, we consider the practical case when ground users have different MRRs in general. In this example, we fix $\theta_1=\theta_2=0.4$ and vary $\theta_3$ and $\theta_4$ by assuming they are equal  to show their effects on the proposed UAV trajectory and achievable max-min throughput.  First, from Fig. \ref{hete1_trj}, we observe that as the MRRs of users 3 and 4 become larger, the UAV tends to adjust its trajectory to get closer to these two users so as to meet their increasingly more stringent minimum-rate requirements.
Meanwhile, since the MRRs of users 1 and 2 are fixed, their priority  is higher/lower than that of users 3 and 4 in  Fig. \ref{hete1_trj} (a) and Fig. \ref{hete1_trj} (c)/(d), respectively, in the optimized UAV trajectory design.  This thus results in asymmetric UAV trajectories for the four users in these cases, which is in contrast to the symmetric trajectory in  Fig. \ref{hete1_trj} (b) with the homogeneous users' MRR constraint. As such, users 1 and 2  in Fig. \ref{hete1_trj} (c) and (d) on average have worse air-to-ground channels than users 3 and 4 along the UAV trajectory. Since the same average throughput needs to be achieved for all users, the UAV  needs to lower its speed and even hover at some positions near users 1 and 2 to compensate their inferior channels caused by the asymmetric UAV trajectory.
Similar to Fig. \ref{homo1_rate}, we compare in Fig. \ref{hete1_rate} the max-min throughput  versus the MRRs of users 3 and 4 under the aforementioned four UAV trajectories. As $\theta_3$ and $\theta_4$ increase, the max-min average throughput is observed to decrease gradually, similarly as  in Fig. \ref{homo1_rate}. However, different from Fig. \ref{homo1_rate}, since only two users, namely, users 3 and 4, rather than all users,  increase their MRRs, the decreasing of the max-min throughput  versus $\theta$ in Fig. \ref{hete1_rate} is less significant as compared to that in Fig. \ref{homo1_rate}. 

%

\section{Conclusions}
Motivated by employing UAVs to provide both delay-constrained and delay-tolerant services in future wireless networks, we consider in this paper a UAV-enabled OFDMA network for serving multiple ground users with heterogeneous communication delay requirements.  Specifically, by taking into account the users' MRR constraints, the system max-min average  throughput is maximized via jointly optimizing the UAV trajectory and OFDMA resource allocation.  We first show that the max-min throughput in general decreases with the MRR  of any user, which implies that the throughput gain arising from the UAV's mobility is less significant as the user delay requirements become more stringent.
  Then, we show that directly applying the conventional block coordinate descent method to solve the formulated problem may fail to update the UAV trajectory effectively. To overcome this issue, we propose a new parameter-assisted block coordinate descent algorithm which is shown by simulations to perform satisfactorily.
   Simulation results are provided to characterize the fundamental tradeoff between the system throughput  and the user communication delay under different MRR setups.
     We hope that the results of the paper can help assessing more practically the performance of UAV-enabled wireless systems with heterogeneous delay requirements.
{ Although a single UAV is considered in this work,  the multi-UAV scenario with potentially co-channel interference  \cite{JR:wu2017joint} is worth pursuing.  In addition, a cross-layer design by considering the UAV's finite data  buffer and the queuing delay along with UAV trajectory and physical-layer resource allocation is worthy of further investigation. Finally, as the UAV's energy consumption and hence endurance are crucial  in practice,  it is appealing to consider energy-efficient UAV trajectory design  \cite{zeng2016energy,yang2017energy} under heterogenous user delay requirements.}

\appendices

\bibliographystyle{IEEEtran}
\bibliography{IEEEabrv,mybib}

\end{document}